\documentclass[reqno,12pt]{preprint}
\usepackage[full]{textcomp}

\usepackage[margin=0.5cm,font=small]{caption}

\pdfoutput=1 %for ArXiV and microtype 
\usepackage[expansion=true]{microtype}  
\usepackage{times}
\usepackage{tikz}

\usepackage{comment}

\usepackage{amsmath}
\usepackage{amssymb}
\usepackage{amsfonts}
\usepackage{hyperref}
\hypersetup{
    colorlinks=true,
    linkcolor=black,
    citecolor=black,
    filecolor=black,
    urlcolor=black,
}

\usepackage{breakurl}
\usepackage{enumitem}

\usepackage{mhequ} 
\usepackage{mhsymb} 
\usepackage{mhenvs}

\colorlet{darkblue}{blue!90!black}
\colorlet{darkgreen}{green!40!black}

\def\Label#1{}

\tikzset{
	lbl/.style={
		label={[shift={(0.15,-0.1)}]\scriptsize #1}               
		},
	>=stealth,
	lblt/.style={
		label={[shift={(0,-0.08)}]\scriptsize #1}               
		},
	lblb/.style={
		label={[shift={(0,-0.8)}]\scriptsize #1}               
		},
	lbll/.style={
		label={[shift={(-0.35,-0.45)}]\scriptsize #1}               
		},
	lblr/.style={
		label={[shift={(+0.35,-0.44)}]\scriptsize #1}               
		},
	lbltr/.style={
		label={[shift={(+0.25,-0.15)}]\scriptsize #1}               
		}
}

\usetikzlibrary{patterns,arrows,decorations.pathreplacing,decorations.pathmorphing,snakes,shapes.misc}

\def\supp{\mathop{\mathrm{supp}}}

\newtheorem{assumption}[lemma]{Assumption}

\newtheorem{innercustomgeneric}{\customgenericname}
\providecommand{\customgenericname}{}
\newcommand{\newcustomtheorem}[2]{%
  \newenvironment{#1}[1]
  {%
   \renewcommand\customgenericname{#2}%
   \renewcommand\theinnercustomgeneric{##1}%
   \innercustomgeneric
  }
  {\endinnercustomgeneric}
}
\newcustomtheorem{custcondition}{Condition}

\def\CB{{\mathcal B}}
\def\E{\mathbf{E}}
\newcommand{\ind}{\mathbf{1}}

\let\ssref=\ref
\def\fref#1{Figure~\ssref{#1}}

\def\cref#1{Condition~\ssref{#1}}
\def\Cref#1{Corollary~\ssref{#1}}
\def\eref#1{(\ssref{#1})}
\def\sref#1{\textsection\ssref{#1}}
\def\lref#1{Lemma~\ssref{#1}}
\def\rref#1{Remark~\ssref{#1}}
\def\tref#1{Theorem~\ssref{#1}}
\def\dref#1{Definition~\ssref{#1}}
\def\pref#1{Proposition~\ssref{#1}}
\def\aref#1{Assumption~\ssref{#1}}

\def\real{{\R}}
\def\integers{{\mathbf N}}

\def\argcdot{{\,\cdot\,}}
\newcommand{\cP}{{\ensuremath{\mathcal P}} }
\newcommand{\cT}{{\ensuremath{\mathcal T}} }

{\theorembodyfont{\rmfamily}\newtheorem{example}[lemma]{Example}}
\let\kappa=\varkappa
\let\theta=\vartheta

\date{}

\begin{document}
\title{Non-Equilibrium Steady States for Networks of Oscillators}
\author{No\'e Cuneo\inst{1}, Jean-Pierre Eckmann\inst{2}, Martin Hairer\inst{3}, Luc Rey-Bellet\inst{4}}
\institute{D\'epartement de math\'ematiques, Universit\'e de Cergy-Pontoise, CNRS UMR 8088
\and
D\'epartement de Physique Th\'eorique and Section de Math\'ematiques, University of Geneva
\and Department of Mathematics, Imperial College London
\and Department of Mathematics and Statistics, University of Massachusetts Amherst
}
%\titleindent=0.55cm

\maketitle

\begin{abstract}
Non-equilibrium steady states for chains of oscillators (masses) connected by harmonic and anharmonic springs and interacting with heat baths at different temperatures have been the subject of several studies.  In this paper, we show how some of the
results extend to more complicated networks.
We establish the existence and uniqueness of the non-equilibrium steady state, and show that the system converges to it at an exponential rate. The arguments are based on controllability and conditions on the potentials at infinity.
\end{abstract}

\section{Introduction} \Label{sec:intro}
 
The aim of this paper is to state and prove an extension of the results
of \cite{EPR,ReyTho02ECT,carmona_2007} to the multidimensional case. 
We consider a network of masses connected with springs (interaction
potentials), where some of
the masses interact with stochastic heat baths which can have different
temperatures. We also let each mass interact with a substrate
through some pinning potential.
 We will show that under conditions spelled out in
this paper, any such system has a unique \emph{non-equilibrium stationary state} (invariant measure). We show, moreover,
that the convergence to the steady state is exponential.
The proof follows in principle the ideas of \cite{EPR,EckHai}, but 
the controllability argument uses the more general conditions of \cite{Control},
and the compactness part relies on a Lyapunov function argument similar to \cite{ReyTho02ECT,carmona_2007}.

The new aspects of this paper are twofold: First, we deal with networks of 
springs connecting the masses, and not just with 
1-dimensional chains.
Second, we correct an oversight of \cite{ReyTho02ECT,carmona_2007} (see \rref{rem:liglpnoncompact})
by a careful analysis of the interplay between the coupling potentials, which hold the 
system together, and the pinning potentials, which prevent it from ``flying 
away''. This will require decomposing the phase space into two 
regions, depending on whether the pinning forces or the interaction forces dominate.
In the process, we also obtain sharper estimates on the rate of energy dissipation (see \rref{rem:sharper}).

The conditions on the system come in the following flavors:
\begin{enumerate}
  \setlength{\itemsep}{1pt}
  \setlength{\parskip}{0pt}
  \setlength{\parsep}{0pt}
  \item [C1]: The masses are sufficiently connected to the heat baths.
     \item [C2]: The interaction potentials are non-degenerate.
       \item [C3]: The potentials are homogeneous at infinity and coercive.
       \item [C4]: The limiting interaction forces are locally injective.
	    \item [C5]: The interaction potentials grow at least as fast as the pinning potentials.

\end{enumerate}

We will make C1--C5 precise in the next section. 
C1--C2 will be required to show the uniqueness of the steady state, and actually C1 will have to be more specific
than what common sense would seem to dictate.
C3--C5 will be further required for existence and exponential convergence.

As was shown in \cite{EPR, EPR2}, it is useful to assume that all
potentials are quadratic (at least at infinity). These
results have been extended 
in \cite{EckHai,RBT,carmona_2007} to potentials of polynomial growth subject to C5.

Without C5, decoupling phenomena (related to ``breathers'') may lead to subexponential convergence
to the invariant measure (and much more difficult proofs). In fact, the existence of the invariant measure
when the pinning potentials grow faster than the interaction potentials
has only been obtained for a chain of 3 masses so far (see the extensive
discussion in \cite{HaiMat}), and for some closely related chains of 
rotors, which correspond to the ``infinite pinning'' limit in a sense (see \cite{CEP_3rotors,CE_4rotors,CP_lowerbound}).

The paper is organized as follows. In \sref{sec:setup} we give the precise definitions of the conditions C1--C5 above and state the main result about existence and uniqueness of the invariant measure, and exponential convergence. The proof relies on two main ingredients: (1) H\"ormander's bracket condition, and  (2) a Lyapunov condition on the energy. 
In \sref{sec:genHam} we prove that these two ingredients lead to the desired result, and there we consider more general thermalized Hamiltonian systems (of which networks of oscillators are a special case). Finally, we check that under  C1--C5, networks of oscillators indeed satisfy H\"ormander's condition (\sref{s:hypo}) and the Lyapunov property (\sref{sec:Lyapunov}).

While the discussion in \sref{sec:genHam} is rather standard, the proofs in \sref{s:hypo} and \sref{sec:Lyapunov} are quite specific to our setup; the main technical difficulty there is that the heat baths do not act on all the oscillators directly, so that propagation within the network has to be carefully studied. This difficulty was already present, to a lesser extent, in the works on chains of oscillators mentioned above.

Finally, although we restrict ourselves to smooth potentials here, we mention that systems of particles with singular interactions (but with heat baths acting on all particles) have attracted significant attention lately (see for example \cite{Conrad2010,MR3717917,HerzogMattinglyErg17,Grothaus2015}).

\section{Setup and results}\label{sec:setup}

We consider a finite set $\CG$ of masses.
We denote by $q_v \in \real^n$
and  $p_v \in \real^n$ the position and momentum of each mass $v\in \CG$ (we assume $n\geq 1$).
The phase space is then $\Omega \equiv \real^{2|\CG| n}$, and we write $z = (p,q) = ((p_v)_{v\in \CG}, (q_v)_{v\in \CG})$.

We then introduce a set $\CE \subset \CG \times \CG$
of edges representing the springs, and consider Hamiltonians of the form
\begin{equ}\label{e:hamil}
 H(p,q) = \sum_{v \in \CG} \Bigl(\frac{p_v^2} 2 + U_v(q_v)\Bigr) +
 \sum_{e \in \CE} V_e(\delta q_e)~,
\end{equ}
where the functions $U_v$ are {\em pinning} potentials, the functions $V_e$
are {\em interaction} potentials, and where for $e=(v,v') \in \CE$
we write $\delta q_e = q_{v'}-q_v \in \real^n$.

We view $(\CG, \CE)$ as an undirected graph with
no loop (\ie, no edge of the kind $(v,v)$).
Since the edges $e = (v,v')$ and $\bar e = (v',v)$ are identified, we also adopt the convention that $V_{e}(q_{v'}-q_{v})$ and $V_{\bar e}(q_{v}-q_{v'})$ are equal and both express just one interaction, which appears only once in \eref{e:hamil}.

We now choose a subset $\CB\subset\CG$ of vertices where thermal baths act,
and for every $b\in \CB$ we assume that some temperature $T_b>0$ and
some coupling constant $\gamma_b>0$ are given. For $v\notin \CB$ we
set, for convenience, $\gamma_v = T_v =0$.
With this notation, our model is described by the system of stochastic differential
equations (one equation per $v\in\CG$):
\begin{equ}[e:SDE]
dq_v = p_v\,dt\;,\qquad dp_v = -\nabla_{q_v} H(p,q)\,dt - \gamma_v p_v\,dt + \sqrt{2T_v \gamma_v}\,dW_v(t)\;,
\end{equ}
where the $W_v$ are mutually independent standard $n$-dimensional Wiener processes. 
Note that for $v\notin \CB$, the last two terms in \eref{e:SDE} are absent. We denote by $z_t = (p_t, q_t)$ the solution of \eref{e:SDE}.
For each fixed initial condition $z\in \Omega$, we denote by $\P_z$ the probability distribution of the solutions to \eqref{e:SDE}, and by $\E_z$ the corresponding expectation. We also introduce the transition kernels $P_t(z, \argcdot)$ defined for all $z\in \Omega$, $t\geq 0$, and all Borel sets $A\subset \Omega$ by
\begin{equ}\label{eq:transitionkernel}
P_t(z, A) = \P_z\{z_t \in A\}~.	
\end{equ}

The Langevin heat baths used in  \eqref{e:SDE} are slightly simpler than those in \cite{EPR,EckHai,ReyTho02ECT}. There, the oscillators interact with some classical field theories which are initially Gibbs-distributed, and the (linear) coupling between the oscillators and the fields is chosen so that the latter can be integrated out. The resulting dynamics is similar to \eqref{e:SDE}, but instead of directly acting on the momenta as in \eqref{e:SDE}, the noise and dissipation act on some auxiliary variables which in turn interact with the momenta. The choice of Langevin heat baths (also made in \cite{carmona_2007}) is only for convenience, and the present analysis is easily transposed to the setup of \cite{EPR,EckHai,ReyTho02ECT}. 

We now make C1--C5 precise. We start with C1 in \sref{s:control}, which is in particular satisfied if the network is a chain with heat baths at both ends. In \sref{s:potential} and \sref{s:nearly}, we discuss C2--C5. An example of potentials satisfying C2--C5 that the reader might want to have in mind is\footnote{Throughout the paper, $\|\argcdot\|$ denotes the Euclidean norm.} $V_e=(1+\|\argcdot\|^2)^{\ell_i/2}$ and $U_v = (1+\|\argcdot\|^2)^{\ell_p/2}$, where $\ell_i, \ell_p \in \real$ satisfy $\ell_i \geq \ell_p \geq 2$. (If $\ell_i$ and $\ell_p$ are even numbers subject to the same condition, then one may also take $V_e=\|\argcdot\|^{\ell_i}$ and $U_v = \|\argcdot\|^{\ell_p}$.)

\subsection{Controllability through the springs}\label{s:control}

The following definition is useful: Let $B$ be a subset of $\CG$. 
We say that $B$ is \emph{nicely connected} to $v\in\CG\setminus B$ if 
there exists a vertex $b \in B$ and an edge of the form $(b,v)\in \CE$,
and there is \emph{no other edge} from $b$ to
  $\CG\setminus B$. We define $\CT B$ as the union of $B$ with its
nicely connected vertices in $\CG\setminus B$ (see \fref{fig:nicelyconnected}).
We denote by $\cT^2 \CB, \cT^3 \CB, \dots$ the iterates of this construction.

\begin{figure}[ht]
\begin{center}
\begin{tikzpicture}[line width=0.6pt, scale=1.1] 
 \tikzstyle{every node}=[circle,draw=black, fill=gray!60,inner sep=3pt,line width=0.8pt]

\node[black,lblr=$\mathrm a$] (a1) at (0.0,0) {};
\node[black,lblr=$\mathrm b$] (a2) at (0.8,0.4) {};
\node[black,lblr=$\mathrm c$] (a3) at (0.0,0.8) {};

\node[] (b1) at (1.8,-0.5) {};
\node[] (b2) at (2,1) {};
\node[] (b3) at (1,1.7) {};

\node[lbll=$\mathrm d$] (c1) at (-1,-0.5) {};
\node[lbll=$\mathrm e$] (c2) at (-.7,1.7) {};

\node[] (d2) at (-1.2,0.8) {};

\draw[] (0.3,0.4) ellipse (30pt and 30pt);

\draw (a1) -- (a2);
\draw (a1) -- (a3);

\draw (a2) -- (b1);
\draw (a2) -- (b2);
\draw (a2) -- (b3);
\draw (b2) -- (b3);

\draw (a3) -- (c2);
\draw (a1) -- (c1);
\draw (d2) -- (c2);
\draw (d2) -- (c1);

\draw (0.6,-0.35) node[fill=none, draw=none] {{\footnotesize$B$}};

\end{tikzpicture}
  \caption{In this network, if $B=\{\mathrm a,\mathrm b, \mathrm c\}$, then $\CT B = \{\mathrm a,\mathrm b,\mathrm c,\mathrm d,\mathrm e\}$.}
  \label{fig:nicelyconnected}
  \end{center}
\end{figure}
\vspace{-1em}

\begin{definition}
Let $(\CG,\CE)$ be as above. We say that $\CB \subset \CG$ \emph{controls} $(\CG, \CE)$
if there exists $k \ge 1$ such that
$\CT^k \CB = \CG$.
\end{definition}

This allows us to make C1 precise as\footnote{It was brought to our attention that the same condition appears in \cite[Section 2.2]{DymovAsym2017}.}
\begin{custcondition}{C1}\label{cond:spanTTk} The graph is connected and $\CB$ controls $(\CG, \CE)$.
\end{custcondition}

\begin{remark}
Note that connectedness is a trivial restriction, 
for if the graph is not connected the results apply to each connected component separately. Chains with heat baths at both ends (or even at just one end) obviously satisfy C1. So do some finite pieces of regular lattices, see \fref{fig:slabs}.
As some examples in \fref{fig:slabs} illustrate, controllability is,  unfortunately, not a monotone property in $\CE$: Adding edges, \ie,
  more springs, will sometimes improve controllability, and sometimes destroy it.
On the other hand, given $(\CG,\CE)$, controllability is a monotone property in the set 
$\CB$ of ``initially controlled'' nodes.
\end{remark}

\begin{figure}[ht]
\begin{center}

\begin{tikzpicture}[line width=0.6pt, scale=1.1] 
 \tikzstyle{every node}=[circle,draw=black, fill=gray!60,inner sep=3pt,line width=0.8pt]

\begin{scope}[shift={(0,5.2)}]

\node[black,lblt={\bf B}] (1) at (0,0) {};
\node[lblt=1] (2) at (1,0) {};
\node[lblt=2] (3) at (2,0) {};
\node[lblt=3] (4) at (3,0) {};
\node[lblt=4] (5) at (4,0) {};
\node[lblt=5] (6) at (5,0) {};
\node[lblt=4] (7) at (6,0) {};
\node[lblt=3] (8) at (7,0) {};
\node[lblt=2] (9) at (8,0) {};
\node[lblt=1] (10) at (9,0) {};
\node[black,lblt={\bf B}] (11) at (10,0) {};

\draw[->] (1) -- (2);
\draw[->] (2) -- (3);
\draw[->] (3) -- (4);
\draw[->] (4) -- (5);
\draw[->] (5) -- (6);
\draw[->] (7) -- (6);
\draw[->] (8) -- (7);
\draw[->] (9) -- (8);
\draw[->] (10) -- (9);
\draw[->] (11) -- (10);

\end{scope}

\begin{scope}[shift={(0,2.1)}]

\node[black,lbll={\bf B}] (1) at (0,0) {};
\node[lblb=1] (2) at (1,0) {};
\node[lblb=2] (3) at (2,0) {};
\node[lblb=1] (4) at (3,0) {};
\node[black,lblr={\bf B}] (5) at (4,0) {};
\node[black,lbll={\bf B}] (6) at (0,1) {};
\node[lbltr=1] (7) at (1,1) {};
\node[lbltr=2] (8) at (2,1) {};
\node[lbltr=1] (9) at (3,1) {};
\node[black,lblr={\bf B}] (10) at (4,1) {};
\node[black,lbll={\bf B}] (11) at (0,2) {};
\node[lblt=1] (12) at (1,2) {};
\node[lblt=2] (13) at (2,2) {};
\node[lblt=1] (14) at (3,2) {};
\node[black,lblr={\bf B}] (15) at (4,2) {};

\draw[->] (1) -- (2);
\draw[->] (2) -- (3);
\draw[->] (4) -- (3);
\draw[->] (5) -- (4);
\draw[->] (6) -- (7);
\draw[->] (7) -- (8);
\draw[->] (9) -- (8);
\draw[->] (10) -- (9);
\draw[->] (11) -- (12);
\draw[->] (12) -- (13);
\draw[->] (14) -- (13);
\draw[->] (15) -- (14);
\draw (1) -- (6);
\draw (2) -- (7);
\draw (3) -- (8);
\draw (4) -- (9);
\draw (5) -- (10);
\draw (6) -- (11);
\draw (7) -- (12);
\draw (8) -- (13);
\draw (9) -- (14);
\draw (10) -- (15);

\end{scope}

\begin{scope}[shift={(6,2.1)}]

\node[black,lbll={\bf B}] (1) at (0,0) {};
\node[lblb=2] (2) at (1,0) {};
\node[lblb=4] (3) at (2,0) {};
\node[lblb=5] (4) at (3,0) {};
\node[lblb=6] (5) at (4,0) {};
\node[black,lbll={\bf B}] (6) at (0,1) {};
\node[lbltr=1] (7) at (1,1) {};
\node[lbltr=3] (8) at (2,1) {};
\node[lbltr=5] (9) at (3,1) {};
\node[lbltr=6] (10) at (4,1) {};
\node[black,lbll={\bf B}] (11) at (0,2) {};
\node[lblt=2] (12) at (1,2) {};
\node[lblt=4] (13) at (2,2) {};
\node[lblt=5] (14) at (3,2) {};
\node[lblt=6] (15) at (4,2) {};
\draw[->] (1) -- (2);
\draw[->] (2) -- (3);
\draw[->] (3) -- (4);
\draw[->] (4) -- (5);
\draw[->] (6) -- (7);
\draw[->] (7) -- (8);
\draw[->] (8) -- (9);
\draw[->] (9) -- (10);
\draw[->] (11) -- (12);
\draw[->] (12) -- (13);
\draw[->] (13) -- (14);
\draw[->] (14) -- (15);

\draw (11) -- (7);
\draw (12) -- (8);

\draw (1) -- (6);
\draw (2) -- (7);
\draw (3) -- (8);
\draw (4) -- (9);
\draw (5) -- (10);
\draw (1) -- (7);
\draw (2) -- (8);
\draw (6) -- (11);
\draw (7) -- (12);
\draw (8) -- (13);
\draw (9) -- (14);
\draw (10) -- (15);

\end{scope}

\begin{scope}[shift={(6,-3.5)}]

\node[black,lblb={\bf B}] (1) at (0,0) {};
\node[lblb=?] (2) at (1,0) {};
\node[lblb=?] (3) at (2,0) {};
\node[lblb=1] (4) at (3,0) {};
\node[black,lblb={\bf B}] (5) at (4,0) {};
\node[black,lblt={\bf B}] (6) at (0,1) {};
\node[lblt=?] (7) at (1,1) {};
\node[lblt=?] (8) at (2,1) {};
\node[lblt=1] (9) at (3,1) {};
\node[black,lblt={\bf B}] (10) at (4,1) {};

\draw (1) -- (2);
\draw (2) -- (3);
\draw (4) -- (3);
\draw[->] (5) -- (4);
\draw (6) -- (7);
\draw (7) -- (8);
\draw (9) -- (8);
\draw[->] (10) -- (9);
\draw (1) -- (6);
\draw (2) -- (7);
\draw (3) -- (8);
\draw (4) -- (9);
\draw (5) -- (10);

\draw (1) -- (7);
\draw (3) -- (9);
\draw (2) -- (6);
\draw (4) -- (8);

\end{scope}

\begin{scope}[shift={(0,-3.5)}]
\node[black,lblt={\bf B}] (1) at (1,0.5) {};
\node[lblt=?] (2) at (2,1) {};
\node[black,lblt={\bf B}] (3) at (3,0.5) {};
\node[lblb=?] (4) at (2,0) {};
\end{scope}
\draw (1) -- (2);
\draw (2) -- (3);
\draw (1) -- (4);
\draw (4) -- (3);

\begin{scope}[shift={(0.27,-1.5)}]
\node[black,lbll={\bf B}] (1) at (0,0) {};
\node[black,lbll={\bf B}] (2) at (0,1) {};
\node[black,lbll={\bf B}] (3) at (0,2) {};
\node[lblr=1] (4) at (0.866,0.5) {};
\node[lblr=2] (5) at (0.866,1.5) {};
\node[lblr=3] (6) at (0.866,2.5) {};
\node[lblr=6] (7) at (1.732,0) {};
\node[lblr=5] (8) at (1.732,1) {};
\node[lblr=4] (9) at (1.732,2) {};

\node[lblr=7] (10) at (2.598,0.5) {};
\node[lblr=8] (11) at (2.598,1.5) {};
\node[lblr=9] (12) at (2.598,2.5) {};

\node[lblr=12] (13) at (3.464,0) {};
\node[lblr=11] (14) at (3.464,1) {};
\node[lblr=10] (15) at (3.464,2) {};

\end{scope}
\draw (1) -- (2);
\draw (2) -- (3);

\draw[->] (1) -- (4);
\draw (2) -- (4);
\draw[->] (2) -- (5);
\draw (3) -- (5);
\draw[->] (3) -- (6);

\draw (4) -- (5);
\draw (5) -- (6);

\draw[->] (4) -- (7);
\draw[->] (5) -- (8);
\draw[->] (6) -- (9);
\draw (5) -- (9);
\draw (4) -- (8);

\draw (7) -- (8);
\draw (8) -- (9);

\draw[->] (7) -- (10);
\draw (8) -- (10);
\draw[->] (8) -- (11);
\draw (9) -- (11);
\draw[->] (9) -- (12);

\draw (10) -- (11);
\draw (11) -- (12);

\draw[->] (10) -- (13);
\draw[->] (11) -- (14);
\draw[->] (12) -- (15);
\draw (11) -- (15);
\draw (10) -- (14);

\draw[->] (14) -- (13);
\draw[->] (15) -- (14);

\begin{scope}[shift={(5.9,-1.5)}]

\node[black,lbll={\bf B}] (a1) at (0.0,0) {};
\node[lblr=1] (a2) at (0.606217782649,0.35) {};
\node[lblr=1] (a3) at (0.606217782649,1.05) {};
\node[black,lbll={\bf B}] (a4) at (0.0,1.4) {};
\node[black,lbll={\bf B}] (a5) at (0.0,2.1) {};
\node[lblr=1] (a6) at (0.606217782649,2.45) {};
\node[lblr=2] (b1) at (1.2124355653,0) {};
\node[lblr=3] (b2) at (1.81865334795,0.35) {};
\node[lblr=3] (b3) at (1.81865334795,1.05) {};
\node[lblr=2] (b4) at (1.2124355653,1.4) {};
\node[lblr=2] (b5) at (1.2124355653,2.1) {};
\node[lblr=3] (b6) at (1.81865334795,2.45) {};
\node[lblr=4] (c1) at (2.4248711306,0) {};
\node[lblr=5] (c2) at (3.03108891325,0.35) {};
\node[lblr=5] (c3) at (3.03108891325,1.05) {};
\node[lblr=4] (c4) at (2.4248711306,1.4) {};
\node[lblr=4] (c5) at (2.4248711306,2.1) {};
\node[lblr=5] (c6) at (3.03108891325,2.45) {};
\node[lblr=6] (d1) at (3.63730669589,0) {};
\node[lblr=7] (d2) at (4.24352447854,0.35) {};
\node[lblr=7] (d3) at (4.24352447854,1.05) {};
\node[lblr=6] (d4) at (3.63730669589,1.4) {};
\node[lblr=6] (d5) at (3.63730669589,2.1) {};
\node[lblr=7] (d6) at (4.24352447854,2.45) {};
\draw[->] (a1) -- (a2);
\draw (a2) -- (a3);
\draw[->] (a2) -- (b1);
\draw[->] (a3) -- (b4);
\draw[->] (a4) -- (a3);
\draw (a4) -- (a5);
\draw[->] (a5) -- (a6);
\draw[->] (a6) -- (b5);
\draw[->] (b1) -- (b2);
\draw (b2) -- (b3);
\draw[->] (b2) -- (c1);
\draw[->] (b3) -- (c4);
\draw[->] (b4) -- (b3);
\draw (b4) -- (b5);
\draw[->] (b5) -- (b6);
\draw[->] (b6) -- (c5);
\draw[->] (c1) -- (c2);
\draw (c2) -- (c3);
\draw[->] (c2) -- (d1);
\draw[->] (c3) -- (d4);
\draw[->] (c4) -- (c3);
\draw (c4) -- (c5);
\draw[->] (c5) -- (c6);
\draw[->] (c6) -- (d5);
\draw[->] (d1) -- (d2);
\draw (d2) -- (d3);
\draw[->] (d4) -- (d3);
\draw (d4) -- (d5);
\draw[->] (d5) -- (d6);
\end{scope}

\end{tikzpicture}
  \caption{The elements of $\CB$ are labeled by ``B''. The numerical
  label $k$ indicates that the vertex is in $\CT^k \CB$ but not in
  $\CT^{k-1}\CB$ (with $\CT^0\CB=\CB$), and the
    uncontrollable elements are labeled by ``?''. The arrows indicate the growth of
  $\CT^k\CB$ as a function of $k$. The top five networks are controlled by $\CB$, while the bottom two are not.  The example in the lower left corner was used in \cite{ez}.}
  \label{fig:slabs}
  \end{center}
\end{figure}
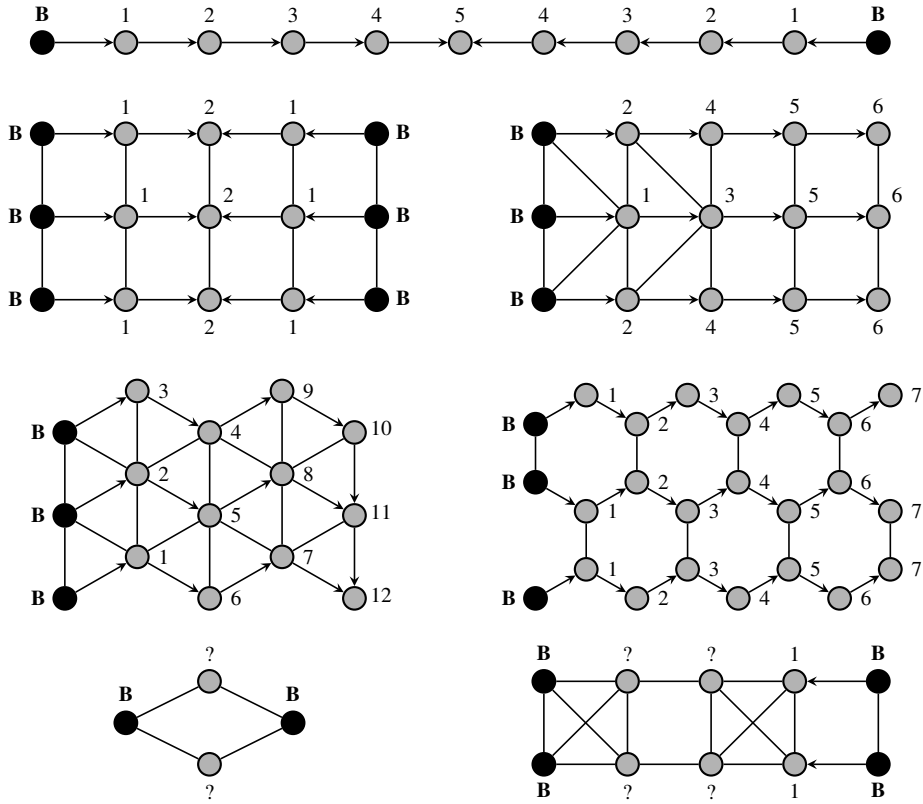

\vspace{-1em}

\begin{remark}\Label{r:1d}
One always has the inequality $|\CT^{k+1} \CB| \le |\CT^{k}\CB| + |\CB|$.
Indeed, let $B_k$ be the set of vertices in $\CT^k\CB$ that are connected to
at least one other vertex in $\CG \setminus \CT^k \CB$.
It is then clear from the definition of $\CT$ that $|\CT^{k+1} \CB| \le |\CT^{k}\CB| + |B_k|$.
On the other hand, it follows from the definition of $\CT$ that,
for every ``newly added'' vertex $v$ in $\CT^{k+1} \CB \setminus \CT^{k}\CB$,
there must be at least one vertex $w$ in $B_k$ such that $v$ is the only element in $\CG\setminus\CT^k \CB$
that is connected to $w$. As a consequence, $|B_{k+1}| \le |B_k| \le |\CB|$ for every $k$, from which the claim
follows at once.

In a way, this remark says that the system is
effectively almost 1-dimensional with respect to the propagation of
information.
No point in $\CB$ and no point in $\CT^k\CB$ will ever
control more than one new point 
as one iterates from $k$ to $k+1$ above.
\end{remark}

Another criterion for the controllability of networks of interacting
oscillators was introduced in \cite{CE_controlling_2014}. While the results
in \cite{CE_controlling_2014} allow in some cases to control networks with more
general topologies, in particular some which do not satisfy \cref{cond:spanTTk},
they only apply to strictly anharmonic
polynomial potentials in 1D ($n=1$).

\subsection{Non-degenerate potentials}\label{s:potential}

We now discuss the conditions on the potentials $V_e$. The attentive reader
will note that, in fact, the non-degeneracy conditions below
are not necessary on all the
links, but only on those which are needed for \cref{cond:spanTTk} to hold.
This means, for example, that in
\fref{fig:slabs},
the potentials associated with the ``vertical'' springs may be degenerate.
We will not deal with this any further, and make the assumptions on
all $V_e$.

Given a multi-index $\alpha =(\alpha _1,\dots,\alpha_n)$  of
non-negative integers, we set $|\alpha |=\sum_{i=1}^n \alpha _i$,
and define $D^\alpha $ as the differential operator with $\alpha _i$
derivatives in the $i^{\rm th}$ direction of $\R^n$.
Given a potential $V\colon \R^n \to \R$, (\ie, any of the $V_e$) we
introduce the following notion of non-degeneracy \cite{ReyTho02ECT}. The idea is that the $V_e$ do not have ``infinitely flat'' pieces.

\begin{definition}\label{d:nondegenerate}
  A smooth potential $V:\R^n\to \R$ is \emph{non-degenerate} if there
  exists an $\ell<\infty$ such that the set of derivatives
$$
\{ D^\alpha  \nabla V (x)~:~ 1 \leq |\alpha| \le \ell \}
$$
spans $\R^n$ for \emph{every} $x\in\R^n$.
\end{definition}

We now have the following precise version of C2:
\begin{custcondition}{C2}\label{cond:nondegVe}
The interaction potentials $V_e$ are non-degenerate.
\end{custcondition}

\begin{example}
Any potential of the form $V(x)=\|x\|^{r}$ with $r=2, 4, 6, \dots$ is non-degenerate. The same is true of $V(x) = (1+\|x\|^2)^{r/2}$  with any real number $r > 0$. On the contrary, if $\|x\|$ is replaced by $|x_1|$ here, then the resulting potential is degenerate (unless $n=1$). 
\end{example}

\begin{remark}
	The condition in \dref{d:nondegenerate} allows for controllability in
the following sense:
Consider a \emph{given} continuous trajectory
 $\bar q \colon [0,1] \to \R^n$ and the problem
\begin{equ}[e:control]
\dot p_f(t) = -\nabla V\bigl(\bar q(t) + f(t)\bigr)\;
\end{equ}
with $p_f(0) = p_*$.
If $V$ is non-degenerate, then the set of solutions $p_f(1)$ of
\eref{e:control} at time $1$,  as $f$ is varied over all smooth
functions with $\sup_{t\le1} |f(t)|\le 1$,
contains an open (and in particular ``full-dimensional'') set. 	
\end{remark}

\subsection{Nearly homogeneous potentials}\label{s:nearly}

One of the difficulties with models of the type \eref{e:hamil}, \eref{e:SDE} is to show the existence
of a non-equilibrium steady state.
As was demonstrated in \cite{HaiMat,Hai}, this can be highly non-trivial, and even with ``nice'' potentials, there are
situations where the convergence to the steady state can be arbitrarily slow.

For the purpose of proving the existence of the steady state, a convenient class of interactions is given by potentials that behave at infinity like homogeneous functions.
We say that a function $\Psi \colon \R^n \to \R$ is homogeneous of degree\footnote{The degree $r$ is not assumed to be an integer. The restriction $r\geq 2$ is required for some of the results below.} $r \geq 2$ if
$\Psi(\lambda x) = \lambda^r \Psi(x)$ for every $\lambda > 0$ and every $x \in \R^n \setminus \{0\}$.
With this notion at hand, we give the following definition, which is slightly weaker than the one
in \cite{RBT}:

\begin{definition}
A smooth function $V\colon \R^n \to \R$ is said to be \emph{nearly homogeneous} of degree $r$ if
there exists a homogeneous (of degree $r$), differentiable function $V_\infty :\real^n \to \real$ such that $\nabla V_\infty$ is locally Lipschitz, and such that for all $0\leq |\alpha| \leq 1$,
\begin{equs}
\lim_{\lambda \to \infty} \sup_{\|x\| = 1}	\left|\frac{(D^\alpha V)(\lambda x)}{\lambda^{r-|\alpha|} } - D^\alpha V_\infty (x)\right| = 0~.
\end{equs}
\end{definition}

\begin{example}
If $V(x) = \|x\|^r$ with $r = 2, 4, 6, \dots$, then $V$ is nearly homogeneous. Moreover, for any real number $r\geq 2$, the potential $V(x) = (1+\|x\|^2)^{r/2}$ is nearly homogeneous. In both cases, $V_\infty(x) = \|x\|^r$.
\end{example}

\begin{remark}\label{replace510}
It is easy to see that nearly homogeneous functions (of degree
$r\ge2$) also satisfy some derived properties, for $0\le|\alpha |\le1$:
\begin{enumerate}
\item[(i)]
$
\lim_{\|x\|\to\infty} \|x\|^{|\alpha|-r}\left({D^\alpha V(x)}-D^\alpha V_\infty(x)\right)= 0$.
\item [(ii)]
$
|D^\alpha V(x)|	\leq C_V (1+\|x\|^{r-|\alpha|})$,
for some $C_V>0$.

\item [(iii)]
$
\lim_{\lambda \to \infty}\sup_{x\in K}	 \left( \lambda^{|\alpha|-r}(D^\alpha
  V)(\lambda x) -D^\alpha V_\infty(x)\right)=
0$, for every compact set $K$.

\item [(iv)]If $\inf_{\|x\| = 1} V_\infty(x)>0$, then 
$
V(x)	\geq C'_V \|x\|^{r}
$,
for some $C'_V>0$ when $\|x\|$ is large enough.

\end{enumerate}
\end{remark}

We can now define C3--C5 properly as follows.

\begin{custcondition}{C3}\label{cond:nearhom}The potentials $U_v$ are nearly homogeneous of degree $\ell_p \geq 2$ with limiting functions $U_{v, \infty}$, and the potentials $V_e$ are nearly homogeneous of degree $\ell_i \geq 2$ with limiting functions $V_{e, \infty}$. Moreover, the limiting potentials are coercive, \ie,  $\inf_{\|x\|=1} V_{e, \infty}(x) > 0$ and $\inf_{\|x\|=1} U_{v,\infty}(x) > 0$. 
\end{custcondition}

\begin{custcondition}{C4}\label{cond:limitingpotentials} The limiting
  interaction forces $-\nabla V_{e, \infty}$ are locally injective in
  the sense that for each  $e\in \CE$ and each $x\in \real^n$, we have
  $\nabla  V_{e, \infty}(x')\neq \nabla  V_{e, \infty}(x)$ for all $x'$ in
  a neighborhood of $x$.
\end{custcondition}

\begin{custcondition}{C5}\label{cond:lilp}The interaction and pinning powers satisfy $\ell_i \geq \ell_p$.
\end{custcondition}

Note that Conditions \ref{cond:nondegVe} and \ref{cond:limitingpotentials} are not comparable: the former guarantees that the forces $-\nabla V_e$ are locally {\em surjective} in a sense, and the latter guarantees that the limiting forces $-\nabla V_{e, \infty}$ are locally {\em injective}.

For example, consider a smooth homogeneous function $V: \real^3 \to \real$ given by $x^4/4 + y^2z^2/2$ on the set $M = \{(x, y, z)\in \real^3~:~z^2+y^2 \leq  x^2/10\}$. Then $\nabla V(x,y,z) = (x^3, yz^2, y^2z)$, and obviously $\{ D^\alpha  \nabla V (x)~:~ 1 \leq |\alpha| \le 3 \}$ spans $\real^3$ for all $(x,y,z)\in M$. However, $\nabla V(x,y,0) = (x^3, 0, 0)$, and thus $\nabla V$ is not locally injective.

There are specific systems for which \cref{cond:limitingpotentials} is not actually required, and
others for which it is, as we
illustrate in Remarks~\ref{rem:waiveCondC4} and \ref{rem:CondC4matters}.

\begin{remark} \cref{cond:limitingpotentials} holds for example if the $V_{e, \infty}$ are strictly convex. In particular if $n=1$, then the $V_{e, \infty}$ are automatically strictly convex, since they are homogeneous of degree $\ell_i \geq 2$ and coercive.
\end{remark}

\begin{remark}The requirement that all interaction potentials have the {\em same} degree $\ell_i$ is crucial. Indeed, if one of the interactions in the bulk of the network (\ie, involving two oscillators in $\CG \setminus \CB$) has a higher degree than the others, the system may find itself in a regime where the two corresponding oscillators oscillate in phase opposition and with a frequency much higher than the other natural frequencies of the system, leading to a decoupling phenomenon comparable to the situation in \cite{HaiMat}. This is again expected to lead to subgeometric convergence to the invariant measure and much more involved proofs.
\end{remark}

\begin{remark}As will be clear from the proofs in \sref{sec:Lyapunov}, it is actually not necessary for all the limiting pinning potentials $U_{v,\infty}$ to be coercive (or even to be non-zero). In fact, we only need the quantity defined in \eqref{eq:Uinfty} to be coercive.
\end{remark}

Without loss of generality, we also assume that the potentials $U_v$ and $V_e$ are non-negative (by the coercivity condition above,
this is always achievable by adding a constant).

\subsection{Main result}

Given the definitions of \sref{s:control}--\ref{s:nearly},
we can now state the main result.  In order to emphasize the role of each assumption,
we introduce the following (very weak) auxiliary condition.
\begin{custcondition}{CA}\label{cond:auxiliary}The Hamiltonian $H$ has compact level sets (\ie, the set $\{z: H(z) \leq K\}$ is compact for each $K>0$), and there exists some $\beta> 0 $ such that the function $\exp\bigl(-\beta H\bigr)$ is integrable on $\Omega$.
\end{custcondition}

\cref{cond:auxiliary} follows immediately from Condition \ref{cond:nearhom} (one can choose any $\beta > 0$).

\begin{theorem}\label{theo:main}The following holds.
\begin{enumerate}
\item 	Under Conditions \ref{cond:spanTTk}, \ref{cond:nondegVe} and  \ref{cond:auxiliary}, the system \eref{e:SDE} admits at most one invariant measure, and if it exists, it has a smooth density with respect to Lebesgue measure.
\item Under Conditions \ref{cond:spanTTk}, \ref{cond:nearhom}, \ref{cond:limitingpotentials} and \ref{cond:lilp}, the system  \eref{e:SDE} admits a least one invariant measure, and $e^{\theta H}$ is integrable with respect to it for all $0 < \theta < 1/T_{\rm max}$, with $T_{\rm max} = \max\{T_b: b\in \CB\}$.
\item Finally, assuming Conditions \ref{cond:spanTTk}--\ref{cond:lilp}, the system \eref{e:SDE} admits a unique invariant measure $\mu_\star$. Moreover, for all $0 < \theta < T_{\rm max}$, there are constants $C,c>0$ such that for every initial condition $z = (p,q) \in \Omega$ and all $t\geq 0$,
\begin{equ}[eq:convergencemu]
\sup_{f\in C(\Omega)\,:\,|f| \leq e^{\theta H}}\left|\E_z f(z_t) - \int f d \mu_\star\right| \leq C e^{\theta H(z)-c t}~.
\end{equ}
\end{enumerate}

\end{theorem}

This theorem is a special case of \tref{theo:maingeneral} below, as we will show.

\section{A general result about thermalized Hamiltonian systems} \label{sec:genHam}

In this section, we prove a version of  \tref{theo:main} which applies to more general thermalized Hamiltonian systems subject to two assumptions H1 and H2 (see below). As we show in \sref{s:hypo} and \sref{sec:Lyapunov}, these assumptions follow from Conditions \ref{cond:spanTTk}--\ref{cond:lilp}. Although the material discussed in this section is mostly standard (see for example \cite{MR1931266}), we provide a complete exposition relying on the version of Harris' ergodic theorem proved in \cite{hairer_yet_2011}. We hope that by considering more general Hamiltonian systems and conditions in this section, the proofs will be both easier to read and useful beyond the scope of this paper.

 The setup is as in \sref{sec:setup}, except that we do not assume that  the set of masses $\CG$ has the structure of a graph and that the Hamiltonian has the form \eqref{e:hamil}. More precisely, we study the SDE
\begin{equ}[e:SDEgen]
dq_v = p_v\,dt\;,\qquad dp_v = -\nabla_{q_v} H(p,q)\,dt - \gamma_v p_v\,dt + \sqrt{2T_v \gamma_v}\,dW_v(t)\;,
\end{equ}
where the friction constants $\gamma_v$, the temperatures $T_v$ and the set  $\CB\subset \CG$ are as in \sref{sec:setup}, and where the Hamiltonian is given by
\begin{equ}
 H(p,q) = \sum_{v \in \CG} \frac{p_v^2} 2 + U(q)~,
\end{equ}
for some arbitrary smooth, non-negative potential $U$ on $\real^{n|\CG|}$.

We also assume throughout this section that \cref{cond:auxiliary} holds, \ie, that $H$ has compact level sets and that $\exp\bigl(-\beta H\bigr)$ is integrable on $\Omega$ for some $\beta > 0$.

We define the semigroup $(\cP^t)_{t\geq 0}$ acting on the space of bounded measurable functions on $\Omega$ by $\cP^t f(z) = \E_z f(z_t) = \int_\Omega f(z') P_t(z, dz')$. 
We also fix $0 < \theta < 1/T_{\rm max}$, with $T_{\rm max}  = \max\{T_b: b\in \CB\}$ as in \tref{theo:main}.  We let moreover $$V = e^{\vartheta H}~.$$

The solutions to \eref{e:SDEgen} form a Markov process whose generator $\CL$ is
\begin{equ}[e:defL]
\CL = X_0 + \sum_{b \in \CB} \sum_{i=1, \dots, n}X_{b,i}^2\;,
\end{equ}
where $X_{b,i} = \sqrt{T_b \gamma_b} \d_{p^i_b}$ and 
\begin{equ}
X_0 = \sum_{v\in \CG} \Bigl(p_v \cdot \nabla_{q_v} - \nabla_{q_v} U(q) \cdot \nabla_{p_v} - \gamma_v p_v \cdot  \nabla_{p_v}\Bigr) \;.
\end{equ}
From now on, we will view $X_0$ and the $X_{b,i}$ interchangeably as first-order differential operators and as vector fields on $\Omega$.

With $C_* = \theta\sum_{ b\in \CB} \gamma_b T_b$, we obtain
\begin{equ}[eq:lyapunovplus]
\CL V = \sum_{b\in \CB} \theta \gamma_b \left([\theta T_b-1]p_b^2 + T_b\right)e^{\theta H} \leq C_*V~.
\end{equ}
 Since $H$, and hence $V$, have compact level sets by assumption, the process admits strong solutions that are continuous and defined for all $t\geq 0$ (almost surely),  the strong Markov property is satisfied, and for all $t\geq 0$ we have
\begin{equ}[eq:lyapunovplusint]
\cP^t V \leq  e^{C_*t}V
\end{equ}
 (see for example \cite[Theorem 3.5]{Khas}, \cite{RB03EPM}, and \cite[Theorem III.3.1]{revuz_continuous_1999} for the strong Markov claim).  

We now introduce  H\"or\-mander's celebrated ``Lie bracket condition'' \cite{hormander_1967}. Define a family of vector fields $\CA_0$ by $\CA_0 = \{X_{b,i}\,:\, b\in \CB, i=1, \dots, n\}$ and then, recursively,
\begin{equ}
\CA_{k+1} = \CA_k \cup \bigl\{[X,Y]\,:\, X \in \CA_k\,, Y \in \CA_0 \cup \{X_0\}\bigr\}~,
\end{equ}
where $[X,Y]$ denotes the Lie bracket (commutator) of $X$ and $Y$.
With this notation at hand, we formulate

\begin{custcondition}{H1}\label{cond:HHormander} The operator $\CL$ defined in \eref{e:defL} satisfies H\"ormander's bracket condition, \ie, for every $z \in \Omega$,
there exists an integer $k>0$ such that the linear span of $\{Y(z): Y \in \CA_k\}$ is all of $\Omega$.
\end{custcondition}

\cref{cond:HHormander} is sufficient (and ``almost necessary'') for $\partial_t - \CL$ to be hypoelliptic, so that the semigroup associated to \eref{e:SDEgen} has
a smoothing effect (see \pref{prop:regularite} below). We note that the requirement in \cref{cond:HHormander} is made for all $z\in \Omega$; see for example \cite{RaquepasOsc2017} for an argument which only requires H\"ormander's condition to hold at one point, but which is specific to quasi-harmonic systems whose harmonic part is subject to Kalman's controllability condition.

Next, we introduce a Lyapunov condition, which will be crucial in order to obtain the existence of an invariant measure and the exponential convergence \eqref{eq:convergencemu}.

\begin{custcondition}{H2}\label{cond:HLyapunov} There exists $t_* > 0$ and $\kappa \in (0,1)$ such that\footnote{Here and below, $\ind_K$ denotes the characteristic function of the set $K$.}
\begin{equs}\label{eq:flyap}
\cP^{t_*} V \leq \kappa  V+ c\ind_{K}~,
\end{equs}
where $c>0$ is a constant and $K$ is a compact set. 
\end{custcondition}

In \sref{s:hypo}, we show that for the original system \eqref{e:SDE}, Conditions \ref{cond:spanTTk} and \ref{cond:nondegVe}  imply \cref{cond:HHormander}, and in \sref{sec:Lyapunov} we show that Conditions \ref{cond:spanTTk}, \ref{cond:nearhom}, \ref{cond:limitingpotentials} and \ref{cond:lilp} imply \cref{cond:HLyapunov}. With this in mind, \tref{theo:main} is a special case of

\begin{theorem}\label{theo:maingeneral}The following holds (recall that \cref{cond:auxiliary} is assumed throughout this section).
\begin{enumerate}
\item 	Under \cref{cond:HHormander}, the system \eref{e:SDEgen} admits at most one invariant measure, and if it exists, it has a smooth density with respect to Lebesgue measure.
\item Assuming \cref{cond:HLyapunov}, the system  \eref{e:SDEgen} admits a least one invariant measure, and $V$ is integrable with respect to it.
\item Finally, assuming Conditions \ref{cond:HHormander} and \ref{cond:HLyapunov}, the system \eref{e:SDEgen} admits a unique invariant measure $\mu_\star$, and the exponential convergence in \eqref{eq:convergencemu} holds.
\end{enumerate}
\end{theorem}
\begin{proof}
The three parts of the theorem are proved in Propositions~\ref{prop:uniq}, \ref{prop:existence} and \ref{prop:expoconv} below. 	
\end{proof}

\subsection{Controllability and uniqueness}
\label{ssec:control}

The following consequence of H\"ormander's condition is well known \cite{hormander_1967} (see \cite[Section 7]{RB03EPM}, \cite{hairer_malliavins_2011}, and \cite[Section 7.4]{MR2410225} for introductions).

\begin{proposition}\label{prop:regularite}
	Assume \cref{cond:HHormander}. Then the transition kernel in \eqref{eq:transitionkernel} can be written as $P_t(z, dz') = p_t(z,z')dz'$, where the map $(t, z, z')\mapsto p_t(z,z')$ is smooth on $(0,\infty) \times \Omega \times \Omega$. In particular, the process is strong Feller. Finally, every invariant measure has a smooth density with respect to Lebesgue measure on $\Omega$.
\end{proposition}

We now prove the following ``accessibility'' result (see also \cite[Section~5.2.1]{CEP_3rotors} for another variant of this argument).

\begin{proposition}\label{prop:uniq} Assume \cref{cond:HHormander}. Then the system \eref{e:SDEgen} admits at most one invariant measure, and for every non-empty open set $\CU\subset \Omega$ and all $z\in \Omega$, we have $\sup_{t> 0}P_t(z, \CU) > 0$.
\end{proposition}

\begin{proof}
The argument follows the same lines as the reasoning first given in \cite{Control}, see also
\cite{Florian}.
Take $\beta > 0$ as in \cref{cond:auxiliary} and consider instead
of \eref{e:SDEgen} the modified equation
\begin{equ}[e:SDE2]
dq_v = p_v\,dt\;,\qquad dp_v = -\nabla_{q_v} H(p,q)\,dt - \gamma_v p_v\,dt+ \sqrt{2\gamma_v \beta^{-1}}\,dW_v(t)\;.
\end{equ}
The only difference is that all the temperatures have been replaced by $1/\beta$ (we still have $\gamma_v = 0$ for all $v\notin \CB$). By the same argument as above, the solutions to \eqref{e:SDEgen} almost surely exist for all times.
It is well known that the measure $$d\mu_\beta = \frac 1 Ze^{-\beta H(p,q)}\,dp\,dq$$ is invariant for \eref{e:SDE2}, and by \cref{cond:auxiliary}, one can choose $Z>0$ so that $\mu_\beta$ is a probability measure. (The invariance of $\mu_\beta$ can be seen by checking that $\CL^* e^{-\beta H} = 0$, where $\CL^*$ is the formal adjoint of the generator of \eqref{e:SDE2}.)

We  next show that $\mu_\beta$ is the only invariant probability measure for \eref{e:SDE2}. It is easy to show that, as a consequence of \pref{prop:regularite}, the map $z\mapsto \overline P_t(z, \argcdot)$ is continuous in the total variation topology, where $\overline P_t$ denotes the transition probabilities for \eqref{e:SDE2}. Since distinct ergodic invariant probability measures for \eref{e:SDE2}
are mutually singular by Birkhoff's ergodic theorem, this immediately implies that if $\nu$ is an ergodic invariant
measure for \eref{e:SDE2} and $z \in \supp \nu$, then there exists a neighborhood $\CU_z$ of $z$ such that 
$\CU_z \cap \supp \bar \nu = \emptyset$ for \textit{every other} ergodic invariant measure $\bar \nu$. 

As a consequence, let us choose some (there exists at least one) ergodic invariant measure $\nu$  of \eref{e:SDE2}.
Assuming by contradiction that $\nu$ is not unique, we have $\supp \nu \neq \Omega$.
As a consequence, setting $\CV = \bigcup_{z \in \supp \nu} (\CU_z \setminus \supp \nu)$, we have
constructed a non-empty open set $\CV$ such that $\CV \cap \supp \bar \nu = \emptyset$
for \textit{every} ergodic invariant measure $\bar\nu$ of \eref{e:SDE2} and therefore, by the ergodic representation theorem,
for \textit{every} invariant measure $\bar\nu$. (We must have $\CV
\neq \emptyset$ for otherwise $\supp \nu$ would be both open and
closed, which cannot be.) However, $\supp \mu_\beta  = \Omega$, thus yielding a contradiction.

Returning to our main line of argument, since $\mu_\beta$ is the unique invariant probability measure for \eref{e:SDE2},
it must be ergodic. Since $\mu_\beta$ has full support, it then follows from Birkhoff's ergodic theorem that for every open set $\CU$
and Lebesgue-almost every initial condition $z\in \Omega$, we have $\sup_{t> 0}\overline P_t(z, \CU) > 0$. An easy application 
of the Chapman-Kolmogorov equation, using the smoothness of the transition probabilities, shows that this actually holds for every $z\in \Omega$.\footnote{One can also use that the strong Feller property 
implies that Birkhoff's ergodic theorem holds for every initial
condition in the support of the invariant measure \cite[Theorem 4.10]{HSV}.}
The conclusion of the proposition thus holds for \eqref{e:SDE2}. We now return to \eqref{e:SDEgen}.

The key is that for each $z\in \Omega$ and $t\geq 0$, the transition probabilities $\overline  P_t(z, \argcdot)$ for \eqref{e:SDE2} and $P_t(z, \argcdot)$ for \eqref{e:SDEgen} are equivalent, since the two stochastic differential equations differ only by the scaling of the Brownian motions. Thus, we indeed have $\sup_{t> 0}P_t(z, \CU) > 0$ for all $z\in \Omega$ and every non-empty open set $\CU \subset \Omega$. Assume now by contradiction that \eqref{e:SDEgen} admits more than one invariant probability measure. Then by the ergodic decomposition theorem there exist two distinct ergodic measures, which then have distinct supports $S_1$ and $S_2$. By smoothness, there exists a non-empty open set $\CU \subset S_2$, and by taking $z\in S_1$ we find $\sup_{t> 0}P_t(z, \CU) = 0$, which is a contradiction.
\end{proof}

Although this will not be needed, we state the following corollary, which follows from the Stroock--Varadhan support theorem (see \cite{stroock_1972}, and \cite[Theorem 5.b]{MR3581833} for an extension to case of unbounded coefficients).
\begin{corollary}
Assume \cref{cond:HHormander}.
Then, for any starting point 
$z_0 \in \Omega$ and any non-empty open set $\CU \subset \Omega$, there exists a time $t > 0$ and smooth
controls $u_b \colon [0,t] \to \real^n$ for $b\in \CB$ such that the solution at time $t$ to\footnote{The same is true without the dissipative terms $- \gamma_v p_v$, since they can be absorbed into the controls $u_v$ (recall that $\gamma_v = 0$ when $v\notin \CB$).}
\begin{equ}
\dot q_v = p_v\;,\qquad \dot p_v =  -\nabla_{q_v} H(p,q) - \gamma_v p_v + \ind_{\CB}(v)u_v(t)\;, \qquad v\in \CG~,
\end{equ}
with initial condition $z_0$, lies in $\CU$. 
\end{corollary}

\begin{remark}
In the case of chains of oscillators, a stronger controllability argument is used in \cite{EPR2}. The argument given above is ``softer''. As a consequence, it applies to a larger class of Langevin equations, at the expense of
having less explicit control. The argument in \cite{EPR2} actually implies that, in the statement of \pref{prop:uniq}, the quantity $P_t(z, \CU)$ is positive for all $t>0$. 
\end{remark}

\subsection{Minorization}

The next proposition shows that every compact set is {\em small} in the terminology of \cite{MT}. In fact, we show that for each given compact set $C$, the minorization condition holds for {\em all} large enough $t$. In the proof, $p_t(\argcdot,\argcdot)$ is as in \pref{prop:regularite}.

\begin{proposition}\label{prop:minorization}Assume Condition \ref{cond:HHormander}. Then, for every compact set $C$, there exists a time $t_C$ such that for all $t\geq t_C$, there exists a non-negative and non-trivial measure $\nu$ (which may depend on $t$) such that $P_{t}(z, \argcdot) \geq \nu$ for all $z\in C$.
\end{proposition}
\begin{proof}
We start by showing that there exists $z_* \in \Omega$ such that for all $z \in \Omega$, there are $t_\sharp(z)$ and $\delta_z > 0$ satisfying
\begin{equ}\label{eq:aperiod}
	p_t(z', z_*) > 0 \qquad \text{for all $t\geq t_\sharp(z)$ and all $z' \in B(z, \delta_z)$ .}
\end{equ}

First, pick any $z_0 \in \Omega$. We now fix any $z_*$ such that $p_1(z_0,z_*) > 0$. By continuity, there exists $\delta > 0$ such that $\inf_{z\in B(z_0, \delta)}p_{1}(z,z_*) > 0$. By \pref{prop:uniq}, there exists for each $z\in \Omega$ some $t_0(z)$ such that $P_{t_0(z)}(z, B(z_0, \delta)) > 0$. It then follows from the semigroup property 
that $p_{t_1(z)}(z, z_*) > 0$ with $t_1(z) = t_0(z) +1$. Using continuity again, we can choose $\delta_z > 0$ so that
\begin{equ}[eq:part1lemferry]
	p_{t_1(z)}(z', z_*) > 0 \qquad \text{for all } z'\in B(z, \delta_z)~.
\end{equ}
We now show that there exists $t_2 > 0$ such that
\begin{equ}[eq:part2lemferry]
	p_t(z_*, z_*) > 0 \qquad \text{for all $t\geq t_2$ }.
\end{equ}
Since $p_{t_1(z_*)}(z_*, z_*) > 0$, continuity with respect to time implies that for some $\Delta > 0$ small enough, we have $p_t(z_*, z_*) > 0$ for all $t\in [t_1(z_*), t_1(z_*) +\Delta]$. But then the same holds for all $t\in [nt_1(z_*), nt_1(z_*) + n\Delta]$, $n \in \integers$. Thus \eref{eq:part2lemferry} holds with $t_2 = n_*t_1(z_*)$ for any integer $n_* \geq  t_1(z_*)/\Delta$. Using \eref{eq:part1lemferry}, \eref{eq:part2lemferry} and the semigroup property yields  \eqref{eq:aperiod} with $t_\sharp(z) = t_1(z) + t_2$.

We now prove the main claim. Let $C$ be a compact set. The balls $\{B(z, \delta_z): z\in C\}$ form an open cover of $C$, and by compactness we can extract a finite subcover, yielding a maximum time $t_C$ such that $p_{t}(z, z_*) > 0$ for all $z\in C$ and all $t\geq t_C$. For any such $t$, since $p_t(\argcdot,\argcdot)$ is continuous on $\Omega^2$ and $C$ is compact, the result follows with $d\nu = \varepsilon\ind_{B(z_*, r)}dz$ for small enough $\varepsilon, r > 0$.
\end{proof}

\subsection{Existence of an invariant measure and exponential convergence}\label{sec:expoconv}

As an elementary consequence of \cref{cond:HLyapunov}, we find that 
\begin{equ}
	\cP^{nt_*} V  \leq    \kappa^{n}V + c\sum_{i=0}^{\infty}\kappa^{i}\quad \text{for all $n\in \integers$}~.
\end{equ}
From this and \eqref{eq:lyapunovplusint}, we obtain that
\begin{equ}\label{eq:semigrouptcontinu}
	\cP^t V  \leq c_1 +  c_2 \rho^tV \quad \text{for all $t\geq0$}~,
\end{equ}
with $c_1, c_2  > 0$ and $\rho = \kappa^{1/t_*}\in (0,1)$. 
In particular, since $V$ has compact level sets, this implies that for any $z\in \Omega$, the family of probability measures $(P_t(z, \argcdot))_{t\geq 0}$ is tight. Since the process is Feller, the standard Krylov--Bogolyubov construction then implies that for some sequence $t_k$ increasing to infinity, $\frac 1 {t_k} \int_0^{t_k}P_s(z, \argcdot)ds$ converges weakly to some measure which is invariant, and with respect to which $V$ is integrable. We thus obtain

\begin{proposition}\label{prop:existence}Under \cref{cond:HLyapunov}, the process admits an invariant measure $\mu_\star$, and $V$ is integrable with respect to $\mu_*$.
\end{proposition}

Assuming in addition \cref{cond:HHormander} implies that $\mu_\star$ is unique (\pref{prop:uniq}), and we now prove exponential convergence.
\begin{proposition}\label{prop:expoconv}
	Under Conditions \ref{cond:HHormander} and \ref{cond:HLyapunov}, the exponential convergence in \eqref{eq:convergencemu} holds.
\end{proposition}
\begin{proof}
	We will apply the main result of \cite{hairer_yet_2011} to the discrete-time semigroup $(\cP^{nt_0})_{n=0,1,2, \dots}$, for some large enough $t_0>0$.  Let first $R  = 2c_1/(1-\rho)$. Here $c_1, c_2$ and $\rho$ are as in \eqref{eq:semigrouptcontinu}. We then define the compact set $C = \{z: V(z) \leq R\}$. We choose now $t_0 \geq t_C$ with the $t_C$ from \pref{prop:minorization}, and large enough so that $c_2 \rho^{t_0} < \rho$. It follows that $R > 2c_1/(1-c_2 \rho^{t_0})$, so that by \eqref{eq:semigrouptcontinu} the main result of \cite{hairer_yet_2011} applies to $(\cP^{nt_0})_{n=0,1,2, \dots}$. We obtain\footnote{Another way to obtain \eqref{eq:discreteconv} with $t_0 = t_*$ is to use \cite[Theorem~15.0.1]{MT}. Indeed, \eqref{eq:aperiod}  implies that the process is aperiodic, \cref{cond:HLyapunov} provides the required drift condition, and by \pref{prop:minorization} the compact set $K$ in \eqref{eq:flyap} is {\em small} (and hence {\em petite}). An alternative proof of convergence
using quasi-compactness of the semigroup can be found in \cite{RB03EPM}.} that for some $C_0, c_0 > 0$ and all $z\in \Omega$,
\begin{equ}\label{eq:discreteconv}
\sup_{f\in C(\Omega) \,:\, |f| \leq V}\left|\E_z f(z_{nt_0}) - \int f d \mu_\star\right| \leq C_0 V(z)e^{-c_0 {nt_0}}~ \quad \text{for all $n\in \integers$ }.
\end{equ}
For $|f| \leq V$, we define $g(z,t) = \E_z f(z_t)-\int f d \mu_\star$. Decomposing $t = n t_0 + r$ with $n\in \integers$ and $r\in [0,t_0)$, we obtain from the Markov property that
\begin{equs}
\left|g(z,t)\right| & = \left|\E_z g(z_{r},nt_0)\right|\leq C_0e^{-c_0 {nt_0}} \E_z V(z_r) \leq C_0e^{C_*t_0-c_0 {nt_0}}  V(z)~,
\end{equs}
where we have also used \eqref{eq:lyapunovplusint}. This immediately implies \eqref{eq:convergencemu} for some $C,c>0$, and thus the proof is complete.
\end{proof}

\section{Hypoellipticity}\label{s:hypo}

In this section, we prove 
\begin{proposition}\label{prop:Hor} Under Conditions \ref{cond:spanTTk} and \ref{cond:nondegVe}, the system  \eref{e:SDE} satisfies \cref{cond:HHormander}.
\end{proposition}

\begin{proof}
For the system \eref{e:SDE}, the vector field $X_0$ in the decomposition \eref{e:defL} reads
\begin{equs}
X_0 &= \sum_{v\in \CG} \Bigl(p_v \cdot  \nabla_{q_v} - \nabla U_v(q_v)\cdot  \nabla_{p_v} - \gamma_v p_v \cdot \nabla_{p_v}\Bigr)\\
& \qquad  - \sum_{(u,v) \in \CE} \nabla V_{(u,v)}(q_v - q_u) \cdot \bigl(\nabla_{p_v}-\nabla_{p_u}\bigr)\;.
\end{equs}

We will actually prove the following statement, which implies \cref{cond:HHormander}. 
Let $\bar X_0 = \d_t - X_0$ and set 
$\CM_0 = \{\bar X_0\} \cup \{X_{b,i} \,:\, b \in \CB, i=1, \dots, n\}$, which we view as a family of smooth vector fields on $\R^{1+2n|\CG|}$. 
Denote by $\CM$ the smallest set of vector fields containing $\CM_0$ that is closed under Lie brackets and multiplication
by smooth functions.

We will show that $\d_t$, as well as $\nabla_{p_v}$ and $\nabla_{q_v}$ for every $v \in \CG$, all belong to $\CM$. Since $\bar X_0 \in \CM$, it is sufficient to prove the claim about the $\nabla_{p_v}$ and $\nabla_{q_v}$.  (Here and below, what we mean by $\nabla_{p_v} \in \CM$ is that $\partial_{p_v^i}\in \CM$ for all $i=1, \dots, n$, and similarly for $\nabla_{q_v}$.)

Note first that, by the definition of $X_{b,i}$ and $\CM_0$, we have
$\nabla_{p_b} \in \CM$ for all $b \in \CB$. Furthermore,
since
\begin{equ}
{} [\d_{p_v^i}, \bar X_0] = -\d_{q_v^i} + \gamma_v \d_{p_v^i}
\end{equ}
for all $v\in\CG$, it follows that one has the implication
\begin{equ}
\nabla_{p_v} \in \CM \quad\Rightarrow\quad \nabla_{q_v} \in \CM\;.
\end{equ}
By the definition of the notion of $\CB$ controlling $\CG$, the claim
now follows if we can show that, for any set $\CB'\subset \CG$, one
has the implication
\begin{equ}
\nabla_{q_b},\nabla_{p_b} \in \CM \text{ for all } b \in \CB' \quad\Rightarrow\quad \nabla_{p_v} \in \CM \text{ for all } v \in \CT \CB'\;.
\end{equ}
Assume therefore that $\CB'$ is such that $\nabla_{q_b},\nabla_{p_b}$
are in $\CM$ for all $b \in \CB'$.
Note that, for all $i \in \{1,\ldots,n\}$ and every $b\in \CB'$,
\begin{equ}[e:comm]
{}[\d_{q_{b}^i}, \bar X_0] = \bigl(\d_i \nabla U_b\bigr)(q_b)\cdot\nabla_{p_b} - \sum_{e=(b,v) \in \CE_b} \bigl(\d_{i} \nabla V_e\bigr)(\delta q_e) \cdot\bigl(\nabla_{p_v}-\nabla_{p_b}\bigr)\;,
\end{equ}
where we denote by $\CE_b$ the subset of those edges in $\CE$ that are of the form $(b,v)$ for some $v \in \CG$.
Fix now $v \in \CT \CB'\setminus\CB'$. By the definition of $\CT \CB'$, there exists
then $b \in \CB'$ such that $(b,v) \in \CE_b$ and, for every other 
$w$ for which  $(b,w)$ is in $\CE_b$, one has $w \in \CB'$. 
For such a $b\in\CB'$, we conclude that in \eref{e:comm} all the terms but $\bigl(\d_{i} \nabla V_{(b,v)}\bigr)(q_v - q_b) \cdot \nabla_{p_v}$ are of the form $f_u(z)\cdot\nabla_{p_u}$ for some $u\in \CB'$, so that 
\begin{equ}\label{e:step1}
\bigl(\d_{i} \nabla V_{(b,v)}\bigr)(q_v - q_b)\cdot \nabla_{p_v} \in \CM\;.
\end{equ}
By the definition of $\CT\CB'$, this holds for every $v \in \CT
\CB'\setminus\CB'$. We now get rid of the potential term in
\eref{e:step1}. Repeatedly taking Lie brackets with $\d_{q_b^j}$, \eref{e:step1} implies that, for every non-zero multi-index $\alpha$, we have
\begin{equ}\label{e:M}
\bigl(D^\alpha \nabla V_{(b,v)}\bigr)(q_v - q_b)\cdot \nabla_{p_v} \in \CM\;.
\end{equ}
Let now $\ell $ be the value appearing in the non-degeneracy
assumption for $V_{(b,v)}$ and let $M$ be the $n\times n$ matrix-valued function
 whose elements are given by
\begin{equ}
M_{ij}(x) = \sum_{1\le |\alpha| \le \ell }\bigl(D^\alpha \d_i V_{(b,v)}\bigr)(x) \, \bigl(D^\alpha \d_j V_{(b,v)}\bigr)(x)\;, \quad x\in\real^n~.
\end{equ}
It follows from the non-degeneracy assumption that $M$ is invertible for every $x\in \real^n$, so that $M^{-1}_{ij}(x)$ is a smooth function.
An explicit calculation shows, furthermore, that one has the identity 
\begin{equ}
\d_{p_v^j} = \sum_{i=1\vphantom{\ell}}^n \sum_{1 \le |\alpha| \le \ell } M^{-1}_{ij}(q_v - q_b) \bigl(D^\alpha \d_i V_{(b,v)}\bigr)(q_v - q_b) \,\bigl(D^\alpha \nabla V_{(b,v)}\bigr)(q_v - q_b) \cdot \nabla_{p_v} \;.
\end{equ}
From \eref{e:M} and the fact that $M^{-1}_{ij}(q_v - q_b) \bigl(D^\alpha \d_i V_{(b,v)}\bigr)(q_v - q_b)$ is a smooth function,
we deduce that we indeed have $\nabla_{p_v} \in \CM$, thus completing the proof. 
\end{proof}

\section{Lyapunov condition}\label{sec:Lyapunov}

In this section, we show that Conditions \ref{cond:spanTTk}, \ref{cond:nearhom}, \ref{cond:limitingpotentials} and \ref{cond:lilp} imply that the system \eref{e:SDE} satisfies \cref{cond:HLyapunov} above, \ie, that $V = e^{\theta H}$ satisfies the Lyapunov property if  $\theta$ is small enough.

The proof follows the lines of the argument that can be found in \cite{ReyTho02ECT,carmona_2007}. Unfortunately,
these works both contained a gap in the argument, which we presently correct (see \rref{rem:liglpnoncompact}).

We fix $t_* > 0$ and $\theta < 1/T_{\rm max}$ with $T_{\rm max}= \max\{T_b: b\in \CB\}$.  The main result of this section is
\begin{theorem}\label{pr:lyap}
Under Conditions \ref{cond:spanTTk}, \ref{cond:nearhom}, \ref{cond:limitingpotentials} and \ref{cond:lilp}, there is a constant $C_1> 0$ such that for all $z_0$ such that $H(z_0)$ is large enough, we have
\begin{equs}[eq:boundC1C2]
\E_{z_0} e^{\theta H(z_{t_*})- \theta H(z_0)}	\leq e^{- C_1 H(z_0)}~.
\end{equs}
\end{theorem}

\begin{remark}By the coercivity of $H$, the theorem above implies that there exist constants  $\kappa \in (0,1)$ and $c>0$, and a compact set $K$ such that
\begin{equs}
\E_z e^{\theta H(z_{t_*})} \leq \kappa  e^{\theta H(z)} + c\ind_{K}(z)~,
\end{equs}
which is the usual Lyapunov condition used in \cref{cond:HLyapunov}.
\end{remark}

For the remainder of the paper, we assume that Conditions \ref{cond:spanTTk}, \ref{cond:nearhom}, \ref{cond:limitingpotentials} and \ref{cond:lilp} are satisfied.

The central role in the proof of \tref{pr:lyap} will be played by the
{\em dissipation integral}
\begin{equs}\label{eq:dissipation_integral}
\Gamma(t) = \sum_{b\in \CB} \gamma_b \int_0^{t} p_b^2(s) ds~.
\end{equs}
In a nutshell, we will prove \eqref{eq:boundC1C2} by showing that if $H(z_0)$ is large enough, then with very high probability the main contribution to the energy difference $H(z_{t_*}) - H(z_0)$ comes from (minus) the dissipation integral $\Gamma(t_*)$, which, also with very high probability, scales like $H(z_0)$.

In order to do this, we start by partitioning, for each initial condition $z_0 \in \Omega$, the probability space into the following three events:
\begin{equs}
A_1 &= \left\{H(z_s) \in \left[\frac{H(z_0)}2, 2H(z_0)\right] ~\forall s\in[0,t_*]\right\}~,\\
A_2 &= \left\{\inf_{s\in[0,t_*]} H(z_s) < \frac{H(z_0)}2\right\}~,\\
A_3 &= \left\{\sup_{s\in[0,t_*]} H(z_s) > 2H(z_0)\right\}~.
\end{equs}

The event $A_1$ will be the center of most of our analysis. 
The event $A_2$ will be of no trouble, since after getting as low as $H(z_0)/2$, it is unlikely that the energy will increase again to a large value. Finally, the event $A_3$ will be of negligible probability at high energy. 

When the event $A_1$ is realized, we will cut the time interval $[0,t_*]$ into subintervals. The length of each subinterval will depend on the distribution of energy between the {\em interaction} and {\em center of mass} degrees of freedom as follows.

We introduce the center of mass coordinates
\begin{equs}[eq:defCDMcoord]
Q &= \frac 1 {|\CG|}\sum_{v\in \CG} q_v~, \qquad P = \sum_{v\in \CG} p_v~,
\end{equs}
and split the Hamiltonian according to
\begin{equs}[eq:separationHams]
H &= H_c+ H_i~,
\end{equs}
where
\begin{equs}
H_c &= \frac{ P^2}{2|\CG|} + \sum_{v\in \CG} U_v(q_v)~,\\
H_i & = \frac 12 \sum_{v\in \CG} \left({p_v} -\frac P {|\CG|}\right)^2 + \sum_{e\in \CE} V_{e}(\delta q_e)~.
\end{equs}

We then let
\begin{equs}[eq:defstoptx]
\tau(z) =  \begin{cases} \lambda H(z)^{\frac 1{\ell_i} - \frac 12} & \text{if } H_i(z) \geq H(z)/2~,\\
\lambda H(z)^{\frac 1{\ell_p} - \frac 12} & \text{if } H_c(z) > H(z)/2~,
\end{cases}
\end{equs}
where $\lambda > 0$ is arbitrary if $\ell_p > 2$, and subject to the condition $0 < \lambda \leq t_*/2$ if $\ell_p = 2$.
Note that $\tau(z)$ is {\em not} random when $z$ is fixed.

The rationale behind \eref{eq:defstoptx} is simple: when the system is
dominated by the ``internal'' dynamics, the natural time scale is
$H(z)^{ 1/{\ell_i} - 1/2}$. In the opposite case, the time
scale $H(z)^{ 1/{\ell_p} - 1/2}$ of the pinning potentials is
relevant. When $\ell_i = \ell_p$, this distinction of time scales
obviously vanishes.

The following proposition, which we will prove in \sref{sec:interacDom} and \sref{sec:pinningdom}, says that with a very large probability, the average dissipation rate over the time interval $[0, \tau(z_0)]$ is at least some fraction of the initial energy.

\begin{proposition}\label{p:tauxU}
Let
\begin{equ}
\tilde A = \left\{H(z_s) \leq  4H(z_0)~\forall s\in[0,\tau(z_0)]\right\}~.
\end{equ}
Then there exist $\varepsilon, B > 0$ such that for all  $z_0$ with $H(z_0)$ large enough,
\begin{equs}[eq:resultptauxU]
\P_{z_0}\left(\tilde A \cap \left\{\Gamma(\tau(z_0))< \varepsilon H(z_0)\tau(z_0) \right\}\right) \leq e^{-B H(z_0)}~.
\end{equs}
\end{proposition}

For the remainder of this section, we assume that $\varepsilon, B$  are fixed as in \pref{p:tauxU}. 
We start with a corollary of \pref{p:tauxU}, which says that one can basically apply \pref{p:tauxU} to successive time intervals in order to obtain estimates on $\Gamma(t_*)$.
\begin{corollary}\label{cor:multitaun}
There exists $B'>0$ such that for all $z_0$ with $H(z_0)$ large enough,
\begin{equs}[eq:toshowcor54]
\P_{z_0}\left(A_1 \cap \left\{\Gamma(t_*) < \frac {\varepsilon  t_* }4 H(z_0) \right\}\right) \leq e^{- B'H(z_0)}~.
\end{equs}
\end{corollary}
\begin{proof}

Fix $z_0$ and let $E =  H(z_0)$.
Consider the sequence of stopping times
\begin{equ}[eq:stoppingtimes]
\tau_0 =0~, \qquad \tau_{j+1}= \tau_j + \tau(z_{\tau_j})~,
\end{equ}
with $\tau(z)$ for $z\in \Omega$ as in \eref{eq:defstoptx}. We now introduce the random variable
\begin{equs}
J = \sup\{ j:  \tau_j \leq t_*\}~.
\end{equs}

On $A_1$, we have for all $t\leq t_*$ that
\begin{equs}
\lambda (2E)^{\frac 1{\ell_i} - \frac 12} \leq \tau(z_t) \leq  \lambda (E/2)^{\frac 1{\ell_p} - \frac 12}~,
\end{equs}
and hence that
\begin{equs}
J \leq \hat J \equiv \lfloor t_*(2E)^{\frac 12 - \frac 1 {\ell_i}}\lambda^{-1} \rfloor~.
\end{equs}
Moreover, if $E$ is large enough (and in the case $\ell_p = 2$, using that $\lambda \leq t_*/2$), we have on $A_1$ that $J > 0$ and that
\begin{equ}
\tau_{J} > t_* - \tau(z_{\tau_{J}})  \geq  \frac{t_*}2~.
\end{equ}

Consider next the events
\begin{equs}
G_j &=   \{J > j\}\cap \left\{\sum_{b\in \CB} \gamma_b \int_{\tau_j}^{\tau_{j+1}} p_b^2(s) ds < \varepsilon \tau(z_{\tau_j}) H(z_{\tau_j}) \right\}~,\\
G &= \bigcup_{j\geq 0}  G_j =  \left\{ \exists j< J : \sum_{b\in \CB} \gamma_b \int_{\tau_j}^{\tau_{j+1}} p_b^2(s) ds < \varepsilon \tau(z_{\tau_j}) H(z_{\tau_j}) \right\}~.
\end{equs}
We observe that the event $A_1 \cap \{J > j\}$ is a subset of
\begin{equ}
	\tilde A_j \equiv \left\{H(z_{\tau_j}) \geq \frac E 2 \text{ and }  H(z_t) \leq 4H(z_{\tau_j}) ~\forall t\in [\tau_j, \tau_{j+1}]\right\}~.
\end{equ}
Thus, if $E$ is large enough, we find by \pref{p:tauxU} and the strong Markov property that
for all $j\geq 0$,
\begin{equs}
\P_{z_0}(A_1 \cap G_j) \leq e^{-BE/2}~,
\end{equs}
so that
\begin{equs}[eq:probGG]
\P_{z_0}(A_1 \cap G) \leq \sum_{j=0}^{\hat J-1} \P_{z_0}(A_1 \cap G_j) \leq \hat J e^{-BE/2} \leq e^{-B' E}
\end{equs}
if $B' > 0$ is small enough and $E$ large enough. 

We observe next that on $A_1 \cap G^c$ and for all $E$ large enough,
\begin{equs}[eq:decompdiss]
\Gamma(t_*)& \geq \sum_{j=0}^{J-1}\sum_{b\in \CB} \gamma_b \int_{\tau_j}^{\tau_{j+1}} p_b^2(s) ds \geq  \sum_{j=0}^{J-1}\varepsilon \tau(z_{\tau_j}) H(z_{\tau_j}) \\
& \geq \frac{\varepsilon E}2  \sum_{j=0}^{J-1} \tau(z_{\tau_j}) =  \frac{\varepsilon E}2  \tau_{J} \geq \frac{\varepsilon E}4  t_*~.
\end{equs}
Thus, the left-hand side of \eqref{eq:toshowcor54} is bounded by $\P_{z_0}\left(A_1 \cap G\right)$, which by  \eref{eq:probGG} completes the proof.
\end{proof}

\begin{lemma}\label{lem:bornerhoUt}
There are constants $\rho, q >0$ such that for every initial condition $z_0$, every event $A$, and all $t>0$,
\begin{equs}[eq:rel1sq]
\E_{z_0} \left(e^{\theta H(z_{t}) - \theta H(z_0)}\ind_A\right)	 & \leq  e^{C_* t}\left(\E_z (e^{-\rho \Gamma(t) }\ind_A) \right)^{\frac 1 q} \leq e^{C_* t} ~,
\end{equs}
with again $C_* = \theta\sum_{ b\in \CB} \gamma_b T_b$.

\end{lemma}
\begin{proof}
	This proof is as in \cite{ReyTho02ECT,carmona_2007}. By applying the It\^o formula to $H(z_t)$, we find
\begin{equs}
\E_{z_0} \left(e^{\theta H(z_{t}) - \theta H(z_0)}\ind_A\right)	 & = e^{C_* t }\E_{z_0} \left(e^{-\theta  \Gamma(t) + \theta  M_{t} } \ind_A\right)~,
\end{equs}
where
\begin{equs}
M_t & =  \int_0^t \sum_{b\in \CB} \sqrt{2\gamma_b T_b} p_b(s) dW_b(s)~.
\end{equs}
The quadratic variation of $M_t$ satisfies
\begin{equ}[eq:quadvarMt]
~[M]_t =  2\int_0^t \sum_{b\in \CB} \gamma_b T_b p^2_b(s) ds \leq  2 T_{\rm max} \Gamma(t)~.
\end{equ}
Let $p> 1$ be such that $p\theta  < 1/T_{\rm max}$ and let $q$ be such that $\frac 1q + \frac 1p = 1$. By H\"older's inequality,
\begin{equs}
&\E_{z_0} \left(e^{-\theta  \Gamma(t) + \theta  M_{t} } \ind_A\right) = \E_{z_0} \left( e^{-\theta \Gamma(t) + \frac {p\theta^2}2 [M]_t} \ind_A e^{ \theta M_t -  \frac {p\theta^2}2 [M]_t}\right)\\
& \qquad \leq \left(\E_{z_0} e^{-\theta q \Gamma(t) + \frac {qp\theta^2}2 [M]_t} \ind_A \right)^{\frac 1 q}  \left(\E_{z_0} e^{ p \theta  M_t -  \frac {p^2\theta^2}2 [M]_t} \right)^{\frac 1 p}~.
\end{equs}
The expectation in the second bracket in the last line is $\leq 1$, since the exponential there is a  Dol\'eans--Dade exponential, and thus a supermartingale. Finally, by \eqref{eq:quadvarMt} we obtain \eref{eq:rel1sq} with $\rho = \theta q(1-p\theta T_{\rm max}) > 0$.
\end{proof}

\begin{lemma}\label{lem:apriori}
There exists $c>0$ such that for all $z_0$ with $H(z_0)$ large enough,
\begin{equs}[eq:2ndbound]
\P_{z_0}(A_3)\leq e^{-c H(z_0)}~.
\end{equs}
\end{lemma}
\begin{proof}
This is a classical result (see for example \cite{ReyTho02ECT} or the proof of Theorem 3.5 in \cite{Khas}). Observe that by \eref{eq:lyapunovplus},
\begin{equs}
(\partial_t + \CL)(e^{\theta H-C_*t}) = (\CL-C_*) e^{\theta H-C_*t}\leq 0~.
\end{equs}
Consider the stopping time $\sigma = \min(t_*, \inf\{t\geq 0: H(z_t) >  2H(z_0)\})$ (with the convention $\inf \emptyset = +\infty$).	 Then,  $\sigma$ is a bounded stopping time, and we have by Dynkin's formula
\begin{equs}
\E_{z_0} e^{\theta H(z_{ \sigma})-C_* t_*} & \leq 	\E_{z_0}  e^{\theta H(z_{ \sigma})-C_* \sigma} \\
& =  e^{\theta H(z_0)} + \E_{z_0} \int_0^{\sigma} ((\partial_s + \CL)(e^{\theta H-C_*s}))(z_s) d s~.
\end{equs}
As the expectation in the last line is non-positive, we find $\E_{z_0}  e^{\theta H(z_{ \sigma})} \leq 	 e^{C_* t_* +\theta H(z_0)}$, and thus
\begin{equs}
\P_{z_0}(A_3) & = \P_{z_0}\{\sigma <  t_*\} \leq  e^{-2\theta H(z_0)} \E_{z_0} \left(e^{\theta H(z_{ \sigma})} \ind_{\sigma < t_*}\right)\leq e^{C_* t_* - \theta H(z_0)}~,
\end{equs}
where the last inequality uses \eref{eq:rel1sq}.
Thus, choosing $c$ small enough completes the proof.
\end{proof}

We can now give the
\begin{proof}[of \tref{pr:lyap}]

First, we have by \lref{lem:bornerhoUt} and \Cref{cor:multitaun} that if $H(z_0)$ is large enough,
\begin{equs}[eq:pieceA1]
&\E_{z_0} \left(e^{\theta H(z_{t_*}) - \theta H(z_0)}\ind_{A_1}\right)	  \leq  e^{C_* t}\left(\E_{z_0} (e^{-\rho \Gamma(t_*) }\ind_{A_1}) \right)^{\frac 1 q} \\
& \qquad \leq e^{C_* t}\left(e^{- B'H(z_0)} + e^{-\rho \varepsilon t_* H(z_0)/4 } \right)^{\frac 1 q} \leq e^{- cH(z_0)}
\end{equs}
for some small enough $c>0$.
We next work on $A_2$. Consider the stopping time $\sigma = \min(t_*, \inf\{t\geq 0: H(z_t) <  H(z_0)/2\})$ (again with $\inf \emptyset = +\infty$). We have $A_2 = \{\sigma <  {t_*}\}$ and
\begin{equs}
	\E_{z_0} \left(e^{\theta H(z_{t_*})}\ind_{A_2}\right) &\leq e^{\theta H(z_0)/2}\E_{z_0} \left(e^{\theta H(z_{t_*})- \theta H(z_{\sigma})}\ind_{A_2}\right)  \leq e^{\theta H(z_0)/2 + C_* t_*}~,
\end{equs}
where we have used the strong Markov property, \eref{eq:rel1sq}, and the fact that $t_*-\sigma \leq t_*$. But then,
\begin{equs}[eq:pieceA2]
\E_{z_0} \left(e^{\theta H(z_{t_*})- \theta H(z_0)}\ind_{A_2}\right) \leq e^{C_*{t_*}- \theta H(z_0)/2} \leq e^{-c H(z_0)}~,
\end{equs}
if $c>0$ is small enough and $H(z_0)$ is large enough.

Finally, by \lref{lem:bornerhoUt} and \lref{lem:apriori}, we have
\begin{equs}[eq:pieceA3]
\E_{z_0} \left(e^{\theta H(z_{t_*})- \theta H(z_0)}\ind_{A_3}\right) & \leq e^{C_* {t_*} }\left(\P_{z_0} (A_3) \right)^{\frac 1 q} \leq e^{- cH(z_0)}~,
\end{equs}
which has the desired form again. Summing \eref{eq:pieceA1}, \eref{eq:pieceA2} and \eref{eq:pieceA3} completes the proof.
\end{proof}

\begin{remark}\label{rem:sharper} Above, we split the time interval $[0,t_*]$ into many subintervals, and apply \pref{p:tauxU} to each of them. This is what allows us to obtain \eref{eq:boundC1C2}, which is very natural from the dimensional point of view. In comparison,  \cite{ReyTho02ECT,carmona_2007}	 use the same Lyapunov function, but obtain weaker estimates (but still sufficient to obtain exponential convergence in \eref{eq:convergencemu}): the bound obtained in \cite{ReyTho02ECT} is $\E_{z_0} e^{\theta H(z_{t_*})- \theta H(z_0)}	\leq e^{- C_1 H^r(z_0)}$ with $r\in (0,1)$, and in \cite{carmona_2007} it is only shown that $\lim_{\|z_0\|\to\infty} \E_{z_0} e^{\theta H(z_{t_*})- \theta H(z_0)} = 0$.
\end{remark}

It now remains to prove \pref{p:tauxU}. In order to do so, we start with some technical lemmas.

\begin{lemma}\label{lem:lemmecontinuite}Let $r\geq 1$ and let $f:\real^r \to \real^r$ be a locally Lipschitz function.
For $T>0$, let $V \in \CC([0,T], \real^r)$ and consider
\begin{equ}
d x_t = f(x_t)dt + dV(t),\qquad d y_t = f(y_t)dt
\end{equ}
with initial conditions $x_0 = y_0 \in \real^r$.
Then, provided that both $x$ and $y$ exist up to time $T$,
\begin{equ}
\sup_{t\in [0,T]}\|x_t-y_t\|  \leq e^{k_*T}   \sup_{t \in [0,T]}\|V(t)\| \;,\qquad 
k_* = \sup_{t \in [0,T]}\frac{\|f(x_t)-f(y_t)\|}{\|x_t-y_t\|}~,
\end{equ}
with the convention $0/0 = 0$.
\end{lemma}
\begin{proof}
Setting $\Delta_s = \|x_s-y_s\|$, we have
$\Delta_t \leq  \int_0^t k_* \Delta_s ds  	+  \|V(t)\|$
and the result follows from Gronwall's inequality.
\end{proof}

\begin{remark}\Label{rem:afterlemmecontinuite}
We will later use \lref{lem:lemmecontinuite} to show that, after adequate rescaling, \eref{e:SDE} (or a component thereof) converges to a deterministic dynamics at high energy.
\end{remark}

As a consequence of the definition of $H$, \cref{cond:nearhom} and \rref{replace510} (iv), we immediately
obtain
\begin{lemma}
There is a constant $C>0$ such that for all $z\in \Omega$, $v\in \CG$ and $e\in \CE$,
\begin{equ}[eq:scalingvar1]
\|q_v\| \leq C (1+H^{\frac 1 {\ell_p}}(z))~,\quad
\|\delta q_e\| \leq C(1+H^{\frac 1 {\ell_i}}(z))~,\quad
\|p_v\| \leq C H^{\frac 12}(z)~.
\end{equ}
	
\end{lemma}

We are now ready to prove \pref{p:tauxU}. We treat the case where
$H_i(z_0) \geq H(z_0)/2$ in \sref{sec:interacDom} and the case where $H_c(z_0) > H(z_0)/2$ in \sref{sec:pinningdom}.
When $\ell_i = \ell_p$, such a distinction is not necessary and only the analysis in \sref{sec:interacDom} is required.

\subsection{When the interactions dominate}\label{sec:interacDom}
In this subsection, we make
\begin{assumption}\label{ass:interdom} If $\ell_i > \ell_p$, we assume that $z_0 \in \Omega$ is such that $H_i(z_0) \geq H(z_0)/2$.
(If $\ell_i = \ell_p$, we make no such restriction.)
\end{assumption}

We write $E =  H(z_0)$. Consider the rescaled time $\sigma = E^{\frac 12 - \frac 1 {\ell_i}}t$ and the variables
\begin{equs}[eq:rescaletilde1]
\tilde p_v(\sigma) &= E^{-\frac 12}p_v(E^{\frac 1 {\ell_i} - \frac 12}\sigma)~,\\
\tilde q_v(\sigma) &= E^{-\frac 1 {\ell_i}}q_v(E^{\frac 1 {\ell_i} - \frac 12}\sigma)~.
\end{equs}

We write $\tilde z = (\tilde p, \tilde q)$ and $\tilde z_0$ for the rescaled initial condition. We consider times $t\in [0,\tau(z_0)] = [0,\lambda E^{1{/\ell_i} -  1/2}]$, or equivalently $\sigma \in [0, \lambda]$.
Observe that in terms of the rescaled time and variables,  \eref{eq:resultptauxU} reads
\begin{equs}[eq:toshowinteracdom]
\P_{z_0}\left(\tilde A \cap \left\{\int_0^\lambda \sum_{b\in \CB} \gamma_b \tilde p_b^2(\sigma) d\sigma < \varepsilon \lambda\right\} \right)  \leq e^{-BE}~.
\end{equs}
In the remainder of this section, we show that \eref{eq:toshowinteracdom} holds provided $E$ is large enough and $z_0$ satisfies \aref{ass:interdom}.

Introducing
\begin{equation}
\tilde H(p,q) = \sum_{v\in \CG}\frac{p_v^2}2 +  \sum_{v\in \CG} E^{-1} U_v(E^{\frac 1 {\ell_i}}q_v)+ \sum_{e\in \CE}E^{-1}V_e(E^{\frac 1 {\ell_i}}\delta q_e)~,
\end{equation}
it is easy to see that
\begin{equs}[eq:tildesysinterac]
d \tilde q_v&= \tilde p_v d \sigma~,\\
d \tilde p_v &= -(\nabla_{q_v} \tilde H)(\tilde p, \tilde q)d\sigma- E^{\frac 1{\ell_i} - \frac 12 } \gamma_v \tilde p_v d \sigma\\
& \qquad  + E^{\frac 1{2 \ell_i} - \frac 34} \sqrt{2 T_v \gamma_v} d \tilde W_v(\sigma) ~,
\end{equs}
where $
\tilde W_v(\sigma) = E^{-\frac 1{2 \ell_i} + \frac 14} W_v(E^{\frac 1{\ell_i} - \frac 12}\sigma)
$ is again an $n$-dimensional Brownian motion. Clearly, in \eref{eq:tildesysinterac}, the stochastic term vanishes in the limit $E \to \infty$, and so does the dissipative term, except when $\ell_i = 2$.

Observe that when $E\to \infty$, the Hamiltonian $\tilde H$ converges pointwise to
\begin{equation}
\hat H(p,q) = \sum_{v\in \CG}\frac{p_v^2}2 +  \delta_{\ell_i, \ell_p}\sum_{v\in \CG} U_{v, \infty}(q_v)+ \sum_{e\in \CE}V_{e,\infty}(\delta q_e)~,
\end{equation}
where $U_{v, \infty}$ and $V_{e,\infty}$ are defined in \cref{cond:nearhom}.

Moreover, by construction,
\begin{equ}
 	H(z) = E \tilde H(\tilde z)~,
\end{equ}
and in particular, 
\begin{equ}[eq:initenergy1]
\tilde H(\tilde z_0) = 1~.
\end{equ}

We introduce the set
\begin{equs}
\tilde K_E = \{\tilde z: \tilde H(\tilde z) \leq 4\}~.
\end{equs}
On the event $\tilde A$, we have $H(z_t) \leq 4E$ for all $t\in [0, \tau(z_0)]$, and hence also
\begin{equ}
\tilde z_\sigma \in \tilde K_E, \quad 0\leq \sigma \leq \lambda~.
\end{equ}

By \eref{eq:scalingvar1}, there exists $\tilde C > 0$ such that if $E$ is large enough, we have that for all $\tilde z \in \tilde K_E$,
\begin{equ}[eq:scalingvartilde1]
\|\tilde q_v\| \leq   \tilde C E^{\frac 1 {\ell_p} - \frac 1 {\ell_i}}~,\quad
\|\delta \tilde  q_e\| \leq \tilde C ~,\quad
\|\tilde p_v\| \leq  \tilde C ~.
\end{equ}

\begin{remark}\label{rem:liglpnoncompact}
Note that if $\ell_i > \ell_p$, then $\tilde q_v$ may become arbitrarily large when $E$ is large, so that the set $\tilde K_E$ is not bounded uniformly in $E$. In fact, when $\ell_i > \ell_p$, it is {\em not} true that $\sup_{\tilde z \in \tilde K_E} | \tilde H(\tilde z) -    \hat H(\tilde z)|$ goes to zero when $E\to \infty$. Indeed, for all $E$ one can find $\tilde z \in \tilde K_E$ such that all the energy is in the pinning potential, so that $\hat H(\tilde z) = 0$ but $\tilde H(\tilde z) = 1$. This explains why we have to restrict ourselves to initial conditions such that $H_i(z_0) \geq H(z_0)/2$ (which will guarantee that $\hat H (\tilde z_0)$ is not too small), and then treat the opposite case separately in \sref{sec:pinningdom}. This distinction is missing from the proofs in \cite{ReyTho02ECT,carmona_2007}.
\end{remark}

\begin{lemma}\label{lem:convergenceTildeInteracDom}
For all $0\leq |\alpha | \leq 1$ and $e\in \CE$, we have
\begin{equs}[eq:limInteracTilde]
	\lim_{E\to\infty}\sup_{\tilde z \in \tilde K_E} \left|D^\alpha_{\tilde z}\left(E^{-1}V_e(E^{\frac 1 {\ell_i}}\delta \tilde q_e) - V_{e,\infty}(\delta \tilde q_e)\right) \right| &= 0~.\\
\end{equs}
Let $v\in\CG$.
If $\ell_i = \ell_p$ and $0\leq |\alpha | \leq 1$ (case 1) or if $\ell_i >\ell_p$
and $|\alpha |=1$ (case 2), then:
\begin{equs}
\lim_{E\to\infty}\sup_{\tilde z \in \tilde K_E}
        \left|D^\alpha_{\tilde z}\left(E^{-1} U_v(E^{\frac 1
          {\ell_i}}\tilde q_v) - \delta_{\ell_i, \ell_p} U_{v,\infty}(\tilde q_v)\right)\right| &= 0~,\label{eq:limPinningTilde}\\
	\lim_{E\to\infty}\sup_{\tilde z \in \tilde K_E} \left|(D^\alpha \tilde H)(\tilde z) -   (D^\alpha \hat H)(\tilde z)\right| &= 0~.\label{eq:limHTildeequa}
	\end{equs}
\end{lemma}

\begin{proof}
The first identity follows immediately from \cref{cond:nearhom}, \rref{replace510} (iii) and \eref{eq:scalingvartilde1}.
Assume now we are in case 1. By \cref{cond:nearhom} and  \rref{replace510} (iii),
\begin{equ}
	\lim_{E\to\infty}\sup_{\|x\| \leq  \tilde C }  \left|\left(E^{\frac{|\alpha|}{\ell_i}-1} (D^\alpha U_v)(E^{\frac 1 {\ell_i}}x) - (D^\alpha U_{v,\infty})(x)\right)\right| = 0~.
\end{equ}
This together with \eref{eq:scalingvartilde1} proves \eref{eq:limPinningTilde}. 

Assume now we are in case 2. By \eref{eq:scalingvartilde1}, in order to prove \eref{eq:limPinningTilde}, it is enough to show that when  $  |\alpha| = 1$,
\begin{equ}[eq:equivboundlime]
	\lim_{E\to\infty}\sup_{\|x\| \leq  \tilde C E^{\frac 1 {\ell_p} - \frac 1 {\ell_i}}}  \left|E^{\frac{|\alpha|}{\ell_i}-1} (D^\alpha U_v)(E^{\frac 1 {\ell_i}}x) \right| = 0~.
\end{equ}
By \rref{replace510} (ii), there exists a constant $c$ such that $ |(D^\alpha U_v)(E^{\frac 1 {\ell_i}}x)| \leq c(1+E^{ {\ell_p/}{\ell_i} - {|\alpha|/}{\ell_i}}\|x\|^{\ell_p - |\alpha|})$. From this, we obtain that for some $c'>0$, the supremum in \eref{eq:equivboundlime} is bounded above by
\begin{equ}
	 c'(E^{\frac{|\alpha|}{\ell_i}-1} + E^{   \frac {|\alpha|} {\ell_i} -\frac {|\alpha|} {\ell_p}})~.
\end{equ}
Clearly, since  $ |\alpha | = 1 < 2 \leq \ell_p < \ell_i$, the above vanishes when $E\to \infty$, which proves \eref{eq:equivboundlime} and hence \eref{eq:limPinningTilde}.

Finally, in both cases, combining \eref{eq:limInteracTilde} and \eref{eq:limPinningTilde} yields \eref{eq:limHTildeequa} (recalling that the kinetic parts in $\hat H$ and $\tilde H$ are identical).
\end{proof}

We now observe that for all $E$ large enough,
\begin{equ}[eq:hatHTildex0]
\hat H(\tilde z_0) \in [1/4, 2]~.
\end{equ} 
Indeed, if $\ell_i = \ell_p$, this follows from \eref{eq:initenergy1} and \eref{eq:limHTildeequa}. If $\ell_i > \ell_p$, then  \aref{ass:interdom} ensures that $|\sum_{v\in \CG} E^{-1} U_v(E^{ 1/{\ell_i}}q_v)| \leq  1/2$, so that \eref{eq:initenergy1} and \eref{eq:limInteracTilde} indeed imply \eref{eq:hatHTildex0} for $E$ large enough.

Next, \eref{eq:tildesysinterac} can be rewritten as
\begin{equs}[eq:tildesysinteracR]
d \tilde q_v&= \tilde p_v d \sigma~,\\
d \tilde p_v &= -(\nabla_{q_v} \hat H)(\tilde p, \tilde q)d\sigma + \tilde R_v(\tilde q)d \sigma \\
& \qquad - E^{\frac 1{\ell_i} - \frac 12 } \gamma_v \tilde p_v d \sigma + E^{\frac 1{2 \ell_i} - \frac 34} \sqrt{2 T_v \gamma_v} d \tilde W_v(\sigma) ~,
\end{equs}
where $
\tilde R_v(\tilde q) =  -\nabla_{\tilde q_v} (\tilde H(\tilde z) - \hat H(\tilde z))$, which by \lref{lem:convergenceTildeInteracDom} satisfies, regardless of whether $\ell_i > \ell_p$ or $\ell_i = \ell_p$, 
\begin{equ}[eq:tildeRtozero]
\lim_{E\to\infty}	\sup_{\tilde z\in \tilde K_E}\| \tilde R_v(\tilde q_v) \| \to 0~.
\end{equ}

Consider now the deterministic limiting system
\begin{equs}[eq:limitsysteminterac]
d \hat q_v&= \hat p_v d \sigma~,\\
d \hat p_v &= -(\nabla_{q_v} \hat H)(\hat z)d\sigma - \delta_{\ell_i, 2}\gamma_v \hat p_v d \sigma ~,
\end{equs}
with initial condition $\hat z_0 = \tilde z_0$.

\begin{proposition}\label{prop:dissipGCInteracDom}
There is a constant $C>0$ such that for every initial condition $\hat z_0$ such that $\hat H(\hat z_0) \in [1/4, 2]$, the solution of \eref{eq:limitsysteminterac} satisfies
\begin{equs}[eq:int0tauhat]
\int_0^\lambda \sum_{b\in \CB} \gamma_b \hat p_b^2(\sigma) d\sigma \geq C~.
\end{equs}
\end{proposition}
\begin{proof}
We first show that 
\begin{equs}[eq:int0tauhatavant]
\int_0^\lambda \sum_{b\in \CB} \gamma_b \hat p_b^2(\sigma) d\sigma > 0 \quad \text{if }\quad   \hat H(\hat z_0)> 0~.
\end{equs}
Indeed, assume  the left-hand side of \eref{eq:int0tauhatavant} is
zero. Then, for all $b\in \CB$, we have $\hat p_b(\sigma) \equiv  0$
on $[0, \lambda]$. Take now $v \in \CT \CB \setminus \CB
$. There exists then $b\in \CB$ such that $b$ is linked only to $v$
and possibly some vertices in $\CB$. Now, since the masses in $\CB$ do not move,
all forces among them are constant (this applies to both the
interaction forces $-\nabla V_{e, \infty}$ with $e\in \CB \times \CB$ and, if $\ell_i = \ell_p$, to the pinning forces $-\nabla U_{b', \infty}$ with $b'\in \CB$). Thus, since the total
force on $b$ is identically zero, we must have that $\nabla V_{(b,v),
  \infty}(\hat q_b(\sigma)-\hat  q_{v}(\sigma))$ is constant. But then,
by \cref{cond:limitingpotentials}, this means that actually $\hat
q_b(\sigma) -\hat  q_{v}(\sigma)$ is constant, and hence that so is $\hat
q_{v}(\sigma)$. We have thus shown that $\hat  p_v(\sigma) \equiv 0$ for all $v\in \CT
\CB$. Proceeding in the same way, we obtain inductively that the same
holds for all $v$ in $\CT^2 \CB$, $\CT^3 \CB$, etc. Thus, by
\cref{cond:spanTTk}, we eventually obtain that 
no mass moves during
the time interval $[0, \lambda]$. But then we
have $\hat p_v(0)=0$ and $\nabla_{q_v} \hat H(\hat z_0) = 0$ for all $v\in \CG$,
which is only possible if $\hat H (\hat z_0) = 0$, so that \eqref{eq:int0tauhatavant} holds.

We now complete the proof of the proposition using a compactness argument and the fact that the solution of \eref{eq:limitsysteminterac} depends continuously on the initial condition $\hat z_0$. In order to do so, there are two cases to consider.
\begin{itemize}
	\item $\ell_i = \ell_p$. Then, the set $\{\hat z: \hat H(\hat z) \in [1/4,2]\}$ is compact, and hence \eref{eq:int0tauhat} holds for some $C>0$.
\item $\ell_i > \ell_p$. Then, the set $\{\hat z: \hat H(\hat z) \in [1/4,2]\}$ is not compact, since it is invariant under global translations $q_v \mapsto q_v + \rho$, where $\rho$ is any vector in $\real^n$ independent of $v$. But when $\ell_i > \ell_p$, both the dynamics \eref{eq:limitsysteminterac} and the left-hand side of \eref{eq:int0tauhat} are invariant under such translations. Since the set $\{\hat z: \hat H(\hat z) \in [1/4,2]\}$ is compact modulo such translations, we obtain \eref{eq:int0tauhat} for some $C>0$.
\end{itemize}
This completes the proof.
\end{proof}

\begin{remark}\label{rem:waiveCondC4} Note that
  \cref{cond:limitingpotentials}
is only used to prove
  \eref{eq:int0tauhatavant}.
In fact, there are systems for which
  \eref{eq:int0tauhatavant}, and hence all the results in the present
  paper, hold without \cref{cond:limitingpotentials}. For example,
  consider a chain of $N$ oscillators with heat baths at both ends, \ie, $\CG = \{1, \dots, N\}$,
  $\CB = \{1, N\}$ and $\CE= \{(1, 2), (2, 3), \dots, (N-1,
  N)\}$. Let $\ell_i > \ell_p$, so that the limiting system only involves
the interaction potentials. Assume the left-hand side of \eref{eq:int0tauhatavant} is
  zero. Then, on the time interval $[0,\lambda]$, we have $\hat
  p_1(\sigma) \equiv 0$. But then, we must have $\nabla V_{(1,2),
    \infty}(\hat q_2(\sigma)-\hat q_1(\sigma)) \equiv 0$ (unlike in
  the general case, we know here that the constant is zero, since no
  other force acts on the first oscillator). As a consequence, since
  the only stationary point of $V_{(1,2), \infty}$ is at the origin
  (this is true of any coercive, homogeneous function without the need
  for \cref{cond:limitingpotentials}), we must have $\hat q_2(\sigma)
  \equiv \hat q_1(\sigma)$. But then we also have $p_2(\sigma) \equiv
  0$. Continuing like this along the chain, we eventually obtain that
  all the masses stand still, and conclude as above that $\hat  H(\hat
  z_0) = 0$. 
\end{remark}

\begin{remark}\label{rem:CondC4matters} \cref{cond:limitingpotentials} cannot be waived in general. We give here a counterexample in three\footnote{Everything in this example happens in the $Oxy$-plane. The third dimension is necessary only to ensure that the interaction potential is non-degenerate.} dimensions consisting of two oscillators 1 and 2, the first of which is coupled to a heat bath (see \fref{fig:C4matters}). We start with both oscillators at rest at position  $(0,1,0)$ and  $(4,2,0)$ respectively. We assume that $V_{21}(x,y,z) = \frac {y^4}4 + \frac {x^2z^2}2$ when $(x,y,z)\approx (4,1,0)$, and that $U_{1} = U_{2} = U$, where $U(x,y,z) = \frac {x^4+y^4+z^4}4$ when $(x,y,z)\approx(0,1,0)$ and $U(x,y,z) = \frac  { x^4}{64}-\frac  {y^4}{32} + \frac {z^4}4$  when $(x,y,z)\approx(4,2,0)$. The potentials above can be extended to non-degenerate, coercive, homogeneous functions of degree 4 (note in particular that $\hat H = H$). Moreover, \cref{cond:limitingpotentials} is not satisfied, as $\nabla V_{21}(4 + \varepsilon, 1, 0) = (0, 1, 0)$ for all small enough $|\varepsilon|$. In this setup, the initial energy of the system is non-zero, and there exists a finite time interval during which the following happens: oscillator 1 does not move at all (so that \eqref{eq:int0tauhatavant} fails to hold if $\lambda$ is small enough), and the position of oscillator 2 is $(x(t), 2, 0)$, for some decreasing $x(t)$. The interaction force $f_{21}$ (see \fref{fig:C4matters}) remains equal to $(0,1,0)$, and the pinning force $f_1$ acting on oscillator 1 remains equal to $(0,-1, 0)$. During the same time, the pinning force $f_2$ acting on oscillator 2 is equal to $(-\frac {x^3(t)}{16}, 1, 0)$, consistently with the motion described above.
\end{remark}

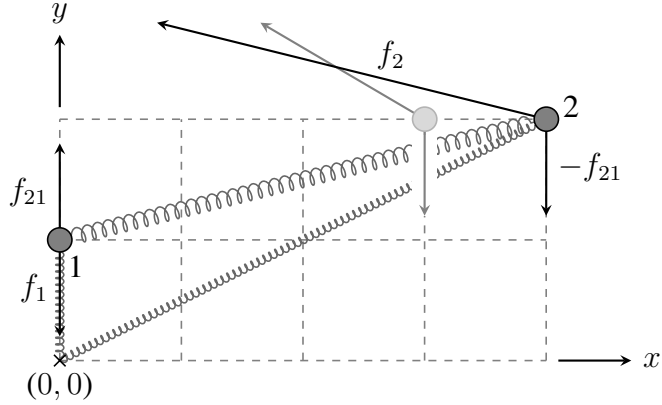
\begin{figure}[ht]
\begin{center}
\scalebox{1.0}{
\begin{tikzpicture}[segment amplitude=15pt, line width=0.6pt, scale=1.6] 

\tikzset{cross/.style={cross out, draw=black, minimum size=1.5*(#1-\pgflinewidth), inner sep=0pt, outer sep=0pt},
%default radius will be 1pt. 
cross/.default={3.5pt}}

\node (x0) at (0,0) {};
\node (x2) at (0,1) {};
\node (x3) at (4,2) {};
\node (x32) at (3,2) {};
%\node (x34) at (3,0.5) {};

\draw[gray,dashed] (0,2)--(4,2);
\draw[gray,dashed] (0,1)--(4,1);
\draw[gray,dashed] (0,0)--(4,0);

\draw[gray,dashed] (0,0)--(0,2);
\draw[gray,dashed] (1,0)--(1,2);
\draw[gray,dashed] (2,0)--(2,2);
\draw[gray,dashed] (3,0)--(3,2);
\draw[gray,dashed] (4,0)--(4,2);

% forces have a factor 0.8

 \draw[snake=coil, segment amplitude=1.5pt, segment length=2.5pt, draw=black!60] (x2)   -- (0,0) node[midway,right] {};

 \draw[snake=coil, segment amplitude=1.5pt, segment length=2.5pt, draw=black!60] (x3)   -- (0,0) node[midway,right] {};

 \draw[snake=coil, segment amplitude=3pt, segment length=4pt, draw=black!60] (x2)   -- (x3) node[midway,below] {};

\fill[white]  (2.9,1.87) rectangle (3.1,1.2);

\draw[->, thick] (x2)++(0,0) -- (0,0.2) node[midway, left] {$f_1$};
\draw[->, thick] (x2)++(0,0) -- (0,1.8) node[midway, left] {$f_{21}$};

\draw[->, thick] (x3)++(0,0) -- (4,1.2)node[midway, right] {$-f_{21}$};
\draw[->, thick, gray] (x32)++(0,0) -- (3,1.2);

\draw[->, thick, gray] (x32)++(0,0) -- (1.65,2.8);

\draw[->, thick] (x3)++(0,0) -- (0.8,2.8)node[pos=0.4, above] {$f_2$};

\draw[->, thick] (0,2.1) -- (0,2.7) node[pos=1, above] {$y$};

\draw[->, thick] (4.1,0) -- (4.7,0) node[pos=1, right] {$x$};

\draw (0,0) node[cross] {} node[below]  {$(0,0)$};

\path[draw=black, fill=gray] (x2) circle (1mm)node[right=6pt, below=1.5pt]{1};

\path[draw=black, fill=gray] (x3) circle (1mm)node[above=5pt, right=2pt]{2};

\path[draw=black!30, fill=gray!30] (x32) circle (1mm)node[below=5pt]{};

\end{tikzpicture}}
  \caption{Illustration of the example in \rref{rem:CondC4matters} in the $Oxy$-plane (in which the motion takes place).}
  \label{fig:C4matters}
  \end{center}
\end{figure}

Returning to the proof of \eqref{eq:toshowinteracdom}, we compare now the systems $(\tilde p, \tilde q)$ and $(\hat p, \hat q)$.
\begin{lemma}\label{lem:convInteracDom}There exist a constant $c>0$, a family of constants $(G_E)_{E > 0} $ satisfying $\lim_{E\to\infty} G_E = 0$, 
and a family of non-negative random variables $(\eta_E)_{E > 0} $ satisfying
\begin{equs}[eq:Rtails]
\P\{\eta_E \geq s\}	\leq e^{-\frac {s^2} 2}~,
\end{equs}
such that if $E$ is large enough,
\begin{equs}
\ind_{\tilde A}\sup_{\sigma\in[0, \lambda]} \|\tilde  z_\sigma - \hat z_\sigma \| \leq G_E + c E^{\frac 1{2 \ell_i} - \frac 34} \eta_E~.
\end{equs}
\end{lemma}
\begin{proof}
The result immediately follows from  \lref{lem:lemmecontinuite} and \eref{eq:tildeRtozero}, provided we can show that there exists an absolute constant $k>0$ such that on the event $\tilde A$, we have $\|(\nabla  \hat H)(\tilde z_\sigma) - (\nabla  \hat H)(\hat z_\sigma) \| \leq k \|\tilde z_\sigma - \hat z_\sigma\|$ for all $0\leq \sigma \leq \lambda$ (we need not worry about the other terms in \eref{eq:limitsysteminterac}, as they are globally Lipschitz).
	As mentioned above, on the event $\tilde A$, we have $\tilde z_\sigma \in \tilde K_E$ for all $0\leq \sigma \leq \lambda$. Moreover, since $\frac d {d \sigma}\hat H(\hat z_\sigma) = - \delta_{\ell_i,2}\sum_{b\in \CB} \gamma_b \hat p_b^2 \leq 0$, we have by \eref{eq:hatHTildex0} that $\hat H(\hat z_\sigma) \leq 2$ for all $0\leq \sigma \leq \lambda$. We consider again two cases separately.

\begin{itemize}
	\item $\ell_i = \ell_p$. Then, there exists $R>0$ such that for all $E$ large enough, $\|\tilde z_\sigma\| \leq R$ and $\|\hat z_\sigma\| \leq R$ for all $0\leq \sigma \leq \lambda$. Since $\nabla \hat H$ is locally Lipschitz (by \cref{cond:nearhom}), the proof is complete.
\item $\ell_i > \ell_p$. Then, one can find $R>0$ such that $\|\delta \tilde q_e(\sigma)\|, \|\delta \hat q_e(\sigma)\| , \|\tilde p_v(\sigma)\|$ and $\|\hat p_v(\sigma)\|$ are bounded by $R$ for all   $0\leq \sigma \leq \lambda$. Since $\nabla \hat H$  is locally Lipschitz and depends only the $\delta q_e$ and $p_v$, the proof is complete.
\end{itemize}

Note that by  \lref{lem:lemmecontinuite}, the random variable $\eta_E$ can be chosen as a constant times $\sup_{\sigma \in [0,\lambda]}\|\tilde W(\sigma)\|$, where $\tilde W = (\tilde W_b)_{b\in \CB}$ is an $n|\CB|$-dimensional Brownian motion. While $\tilde W$ depends on $E$ pathwise, its distribution does not. Moreover, $\tilde W$ does not depend on  $z_0$ for a given energy $E$, and thus the same is true of $\eta_E$.
\end{proof}

Using \lref{lem:convInteracDom}, \pref{prop:dissipGCInteracDom} and the inequality $
x^2  \geq \frac{y^2}2 -	(x-y)^2
$, we obtain that there exist $c, c'>0$ such that on $\tilde A$ and if $E$ is large enough,
\begin{equs}
\int_0^\lambda \sum_{b\in \CB} \gamma_b \tilde p_b^2(\sigma) d\sigma &\geq c -  c'(G_E + E^{\frac 1{2 \ell_i} - \frac 34} \eta_E)^2  \geq c -  2c'G_E^2 - 2 c'E^{\frac 1{ \ell_i} - \frac 32} \eta_E^2~.
\end{equs}
Since $G_E\to 0$, we find for $E$ large enough that
\begin{equs}
\int_0^\lambda \sum_{b\in \CB} \gamma_b \tilde p_b^2(\sigma) d\sigma \geq \frac c 2 - 2c' E^{\frac 1{ \ell_i} - \frac 32} \eta_E^2~,
\end{equs}
so that
\begin{equs}
\P_{z_0}\left(\tilde A\cap \left\{\int_0^\lambda \sum_{b\in \CB} \gamma_b \tilde p_b^2(\sigma) d\sigma < \frac c 4\right\}\right) \leq \P\left\{\eta_E >E^{ \frac 34- \frac 1{2 \ell_i}} \sqrt{\frac { c} {  {8c'}}}\right\}~.
\end{equs}
Using now \eref{eq:Rtails} and the fact that $\frac 1{ \ell_i} - \frac 32 \leq -1$ completes the proof of \eref{eq:toshowinteracdom} (for an adequate choice of $\varepsilon$ and $B$).

Thus, if $\ell_i = \ell_p$, the proof of \pref{p:tauxU} is complete. If now $\ell_i > \ell_p$, then because of \aref{ass:interdom}, the conclusion of \pref{p:tauxU} is proved only in the case where $H_i(z_0) \geq H(z_0)/2$, and the next subsection is required.

\subsection{When the pinning dominates}\label{sec:pinningdom}

Recalling the decomposition of $H$ introduced in \eref{eq:separationHams}, we now make the following assumption.

\begin{assumption}\label{ass:pindom} We assume that $\ell_i > \ell_p$ and that
the initial condition $z_0 \in \Omega$ satisfies $H_c(z_0) > H(z_0)/2$.
\end{assumption}

We start by rescaling the system in much the same way as in \sref{sec:interacDom}, except that we now choose the natural scaling of the pinning.
More precisely, we introduce the rescaled time $\sigma = E^{1/2 -  1 /{\ell_p}}t$ and the variables
\begin{equs}
\tilde p_v(\sigma) &= E^{-\frac 12}p_v(E^{\frac 1 {\ell_p} - \frac 12}\sigma)~,\\
\tilde q_v(\sigma) &= E^{-\frac 1 {\ell_p}}q_v(E^{\frac 1 {\ell_p} - \frac 12}\sigma)~.
\end{equs}
 We consider times $t\in [0,\tau(z_0)] = [0,\lambda E^{\frac 1{\ell_p} - \frac 12}]$, or equivalently $\sigma \in [0, \lambda]$. As in \sref{sec:interacDom}, the analogue of \eref{eq:resultptauxU} in terms of the rescaled variables and time is
\begin{equs}[eq:toshowpinningdom]
\P_{z_0}\left(\tilde A \cap \left\{\int_0^\lambda \sum_{b\in \CB} \gamma_b \tilde p_b^2(\sigma) d\sigma < \varepsilon \lambda\right\} \right)  \leq e^{-BE}~.
\end{equs}

We let now
\begin{equ}
\tilde H(p,q) = \sum_{v\in \CG}\frac{p_v^2}2 +  \sum_{v\in \CG} E^{-1} U_v(E^{\frac 1 {\ell_p}}q_v)+ \sum_{e\in \CE}E^{-1}V_e(E^{\frac 1 {\ell_p}}\delta q_e)~,
\end{equ}
and obtain
\begin{equs}[eq:tildesysintpin]
d \tilde q_v&= \tilde p_v d \sigma~,\\
d \tilde p_v &= -(\nabla_{q_v} \tilde H)(\tilde p, \tilde q)d\sigma- E^{\frac 1{\ell_p} - \frac 12 } \gamma_v \tilde p_v d \sigma + E^{\frac 1{2 \ell_p} - \frac 34} \sqrt{2 T_v \gamma_v} d \tilde W_v(\sigma) ~,\quad 
\end{equs}
where $
\tilde W_v(\sigma) = E^{-\frac 1{2 \ell_p} + \frac 14} W_v(E^{\frac 1{\ell_p} - \frac 12}\sigma)
$ is again an $n$-dimensional Brownian motion.
We define, as in \sref{sec:interacDom},
\begin{equs}
\tilde K_E = \{\tilde z: \tilde H(\tilde z) \leq 4\}~,
\end{equs}
and obtain that on the event $\tilde A$, we have $\tilde z_\sigma \in \tilde K_E$ for all $0\leq \sigma \leq \lambda$.

By \eref{eq:scalingvar1}, there is some $\tilde C$ such that if $E$ is large enough, then for all $\tilde z \in \tilde K_E$,
\begin{equ}[eq:scalingvartildepin]
\|\tilde q_v\| \leq   \tilde C~,\quad
\|\delta \tilde  q_e\| \leq \tilde C  E^{\frac 1 {\ell_i} - \frac 1 {\ell_p}} ~,\quad
\|\tilde p_v\| \leq  \tilde C ~.
\end{equ}

Note that unlike in \sref{sec:interacDom}, the collection of sets $(\tilde K_E)_{E>0}$ is uniformly bounded.
In fact, the maximum allowed value of $\delta \tilde q_e$ becomes very small at high energy. 

\begin{remark}
	The difficulty is that the dynamics \eref{eq:tildesysintpin} does not converge to a nice limit when $E$ is large. Indeed, we have for any edge $e=(v,v')$ that
\begin{equ}
	\nabla_{\tilde q_v} (E^{-1} V_e(E^{\frac 1 {\ell_p}}\delta \tilde q_e)) \sim E^{\frac {\ell_i}{\ell_p}-1} \|\tilde \delta q_e\|^{\ell_i-1}~,
\end{equ}
which diverges pointwise when $E\to \infty$ if $\delta \tilde q_e \neq 0$. The supremum of this quantity over $\tilde K_E$ diverges like 
$E^{1 /{\ell_p} - 1/ {\ell_i}}$ (as can be seen by the scaling in \eref{eq:scalingvartildepin}).
The interpretation is that at high energy and under \aref{ass:pindom}, while the rescaled system behaves like a  ``tight molecule'' 
with vanishing relative distance $\tilde \delta q_e$ between the masses, the dynamics is still dominated by
the fast oscillations of the internal degrees of freedom. The way around this is to consider the center of mass coordinates.
\end{remark}

The center of mass coordinates in \eref{eq:defCDMcoord} are expressed, after rescaling, as
\begin{equs}
\tilde P(\sigma) &= E^{-\frac 12}P(E^{\frac 1 {\ell_p} - \frac 12}\sigma) = E^{-\frac 12} \sum_{v\in \CG} p_v(E^{\frac 1 {\ell_p} - \frac 12}\sigma)~,\\
\tilde Q(\sigma) &= E^{-\frac 1 {\ell_p}}Q(E^{\frac 1 {\ell_p} - \frac 12}\sigma) = \frac 1 {|\CG|} E^{-\frac 1 {\ell_p}} \sum_{v\in \CG} q_v(E^{\frac 1 {\ell_p} - \frac 12}\sigma)~.
\end{equs}
We denote by $(\tilde P_0, \tilde Q_0)$ the rescaled initial condition.
As the interaction forces cancel out, the dynamics we obtain is 
\begin{equs}[eq:PtQtpinfirst]
d \tilde Q &=  \frac 1 {|\CG|} \tilde P d \sigma~,\\
d \tilde P & =-\sum_{v\in \CG}E^{\frac 1{\ell_p}-1}\nabla U_v(E^{\frac 1 {\ell_p}}\tilde q_v) d\sigma  \\
& \qquad  \qquad- E^{\frac 1 {\ell_p} -\frac 12 }\sum_{b\in \CB} \gamma_b \tilde p_b d \sigma +  E^{\frac 1{2 \ell_p} - \frac 34} \sum_{b\in \CB}\sqrt{2 T_b \gamma_b} d \tilde W_b(\sigma) ~.
\end{equs}

Moreover, since the graph $(\CG, \CE)$ is connected by \cref{cond:spanTTk}, we have for all $\tilde z \in \tilde K_E$ that
\begin{equ}[eq:qtildemqtilde]
\max_{v\in \CG }\|\tilde Q-\tilde q_v\|	\leq \max_{(v, v') \in \CG^2} \|\tilde q_v - \tilde q_{v'}\|\leq \sum_{e\in \CE} \|\delta \tilde q_e\| \leq |\CE|\tilde C E^{\frac 1 {\ell_i} - \frac 1 {\ell_p}}  ~.
\end{equ}

Defining now
\begin{equ}\label{eq:Uinfty}
U_\infty (\tilde Q) = \sum_{v\in \CG}U_{v, \infty}(\tilde Q)~,
\end{equ}
we have
\begin{lemma}\label{lem:convergencepotpindom}
For all $0\leq |\alpha| \leq 1$,
\begin{equ}
\lim_{E\to\infty}\sup_{\tilde z \in \tilde K_E}\left| \sum_{v\in \CG}E^{\frac {|\alpha|}{\ell_p}-1}(D^\alpha U_v)(E^{\frac 1 {\ell_p}}\tilde q_v)-(D^\alpha U_\infty) (\tilde Q)\right| = 0~.
\end{equ}
\end{lemma}
\begin{proof}
By \cref{cond:nearhom}, \eref{eq:scalingvartildepin} and \rref{replace510} (iii) we have, for all $v\in \CG$,
\begin{equ}[eq:premierbornelem217]
\lim_{E\to\infty}\sup_{\tilde z \in \tilde K_E}\left|E^{\frac {|\alpha|}{\ell_p}-1}(D^\alpha U_v)(E^{\frac 1 {\ell_p}}\tilde q_v) - (D^\alpha U_{v,\infty})(\tilde q_v)\right| = 0~.
\end{equ}
Moreover, by \cref{cond:nearhom}, there exists $c>0$ such that $D^\alpha U_{v, \infty}$ is $c$-Lipschitz on the ball $B(0, \tilde C)\subset \real^n$, which  by \eref{eq:scalingvartildepin} contains $\tilde q_v$ and $\tilde Q$ for all $\tilde z\in \tilde K_E$. Thus,
\begin{equs}
\sup_{\tilde z \in \tilde K_E}\left|(D^\alpha U_{v,\infty})(\tilde q_v)-(D^\alpha U_{v,\infty})(\tilde Q)\right| \leq c\sup_{\tilde z \in \tilde K_E}\|\tilde Q - \tilde q_v\| ~.
\end{equs}
By \eref{eq:qtildemqtilde}, the right-hand side vanishes when $E\to \infty$. This and \eref{eq:premierbornelem217} imply that
\begin{equ}
\lim_{E\to\infty}\sup_{\tilde z \in \tilde K_E}\left|E^{\frac {|\alpha|}{\ell_p}-1}(D^\alpha U_v)(E^{\frac 1 {\ell_p}}\tilde q_v) - (D^\alpha U_{v,\infty})(\tilde Q)\right| = 0~.
\end{equ}
By the definition of $U_\infty$ and the triangle inequality, the proof is complete.
\end{proof}

It is then natural to consider the limiting system
\begin{equs}[eq:limitsystempinning]
d \hat Q &= \frac 1 {|\CG|}\hat P d \sigma~,\\
d \hat P & =-\nabla U_\infty(\hat Q)d \sigma~,
\end{equs}
which corresponds to the Hamiltonian
\begin{equ}
\hat H(\hat P,\hat Q) = \frac {\hat P^2} {2|\CG|}   + U_\infty (\hat Q)~.
\end{equ}

We can rewrite \eref{eq:PtQtpinfirst} as
\begin{equs}[eq:prelimitsystempinning]
d \tilde Q &=  \frac 1 {|\CG|}\tilde P d \sigma~,\\
d \tilde P & =-\nabla U_\infty(\tilde Q)d \sigma + \tilde R(\tilde z) d\sigma- E^{\frac 1 {\ell_p} -\frac 12 }\sum_{b\in \CB} \gamma_b \tilde p_b d \sigma \\
& \qquad  +  E^{\frac 1{2 \ell_p} - \frac 34} \sum_{b\in \CB}\sqrt{2 T_b \gamma_b} d \tilde W_b(\sigma)~,
\end{equs}
where
\begin{equ}
\tilde R(\tilde z) = \nabla  U_\infty(\tilde Q)- E^{\frac {1}{\ell_p}-1} \sum_{v\in \CG} (\nabla  U_v)(E^{\frac 1 {\ell_p}}\tilde q_v)~,
\end{equ}
which by \lref{lem:convergencepotpindom} satisfies
\begin{equ}[eq:tildeRtozero pinning]
\lim_{E\to\infty}	\sup_{\tilde z\in \tilde K_E}\| \tilde R(\tilde z) \| \to 0~.
\end{equ}

Note that the dynamics \eref{eq:prelimitsystempinning} does {\em not} converge to \eref{eq:limitsystempinning} when $\ell_p = 2$, as the dissipative terms in \eref{eq:prelimitsystempinning} remain in the limit. This will complicate the argument slightly (see the proof of Lemma~\ref{lem:BvarepsilonX}).

As a consequence of \lref{lem:convergencepotpindom}, and since $H_c(z_0) > H(z_0)/2$, we have when $E$ is large enough that
\begin{equ}
\hat H(\tilde P_0, \tilde Q_0) \in [1/4, 2]~.
\end{equ}

\begin{proposition}\label{prop:limitDomBoundedBelow}
There is a constant $C>0$ such that for every initial condition $(\hat P_0,\hat Q_0)$ with $\hat H(\hat P_0, \hat Q_0) \in [1/4, 2]$, the solution of \eref{eq:limitsystempinning} satisfies
\begin{equs}[eq:int0tauhatpin]
\sup_{\sigma \in [0, \lambda]} \|\hat Q(\sigma)-\hat Q(0)\|\geq C~.
\end{equs}
\end{proposition}
\begin{proof}
The left-hand side of \eref{eq:int0tauhatpin} is obviously strictly positive provided that $\hat H(\hat P_0, \hat Q_0) > 0$. Moreover, the map $(\hat P_0, \hat Q_0)\mapsto \sup_{\sigma \in [0, \lambda]} \|\hat Q(\sigma)-\hat Q(0)\|$ is lower semicontinuous (as the supremum of a family of continuous functions). Thus, since the set $\{(\hat P_0, \hat Q_0) \in \real^{2n} :\hat H(\hat P_0, \hat Q_0) \in [1/4, 2]\}$ is compact, the proof is complete.
\end{proof}

We introduce the random variable
\begin{equs}
X = \sup_{b\in \CB} \int_0^\lambda \tilde p_b^2 d \sigma~.
\end{equs}

\begin{lemma}\label{lem:BvarepsilonX}
There exist constants $B,\varepsilon > 0$ such that if $E$ is large enough,
\begin{equs}
\P_{z_0}\left(\tilde A \cap \{X  < \varepsilon \}\right) \leq e^{-BE}~.
\end{equs}	
\end{lemma}
\begin{proof}
In this proof, the constant $c>0$ may be different each time it appears, and is not allowed to depend on $E$, provided $E$ is large enough.
First, observe that for all $b\in \CB$, H\"older's inequality implies that
\begin{equs}[eq:byholder]
	 \int_0^\lambda \|\tilde p_b\|d\sigma \leq c\sqrt{X}~.
\end{equs}

Next, by \eref{eq:scalingvartildepin} and the fact that $\hat H$ is conserved by \eref{eq:limitsystempinning}, we have that on the event $\tilde A$, there is some $R>0$ such that $\|\hat Q(\sigma)\|, \|\tilde Q(\sigma)\|, \|\hat P(\sigma)\|$ and $ \|\tilde P(\sigma)\|$ are bounded by $R$ for all $0\leq \sigma \leq \lambda$. As $\nabla \hat H$ is locally Lipschitz by \cref{cond:nearhom}, there exists $k>0$ such that on the event $\tilde A$, we have for all $0\leq \sigma \leq \lambda$ that
\begin{equ}
\|(\nabla  \hat H)(\hat P(\sigma), \hat Q(\sigma)) - (\nabla  \hat H)( \tilde P(\sigma), \tilde Q(\sigma)) \| \leq k \|(\hat P(\sigma), \hat Q(\sigma))- (\tilde P(\sigma), \tilde Q(\sigma))\|	~.
\end{equ}
As a consequence, we can apply \lref{lem:lemmecontinuite} to  $(\hat P, \hat Q)$ and $(\tilde P, \tilde Q)$ to obtain that
\begin{equs}[eq:resultatdiffpetitere]
\ind_{\tilde A}\sup_{\sigma\in[0, \lambda]} \|\tilde  Q(\sigma) - \hat Q(\sigma) \| \leq c G_E + cE^{\frac 1 {\ell_p} -\frac 12 }\sqrt{X}+ c E^{\frac 1{2 \ell_p} - \frac 34} \eta_E~,
\end{equs}
where $\lim_{E\to\infty}G_E = 0$ and where $\eta_E$ is a non-negative random variable satisfying \eref{eq:Rtails}. 

Pick now any $b\in \CB$. By \eref{eq:tildesysintpin} and \eref{eq:byholder}, we have
\begin{equs}[eq:supsqtv]
	\sup_{\sigma\in [0,\lambda]}\|\tilde  q_b(\sigma) - \tilde q_b(0) \| \leq c \int_0^\lambda \|\tilde p_b\|d\sigma \leq c\sqrt{X}~.
\end{equs}
Moreover, by \eref{eq:qtildemqtilde}, we also find that on $\tilde A$,
\begin{equs}[eq:supqtildeQtildeA1]
	\sup_{\sigma\in [0,\lambda]}\|\tilde  q_b(\sigma) - \tilde Q(\sigma) \| \leq c E^{\frac 1 {\ell_i}-\frac 1 {\ell_p}} ~.
\end{equs}
From \eref{eq:supsqtv} and \eref{eq:supqtildeQtildeA1} we deduce that
\begin{equs}[eq:surQtilde0sigma]
	\sup_{\sigma\in [0,\lambda]}\|\tilde  Q(\sigma) - \tilde Q(0) \| \leq c E^{\frac 1 {\ell_i}-\frac 1 {\ell_p}} + c \sqrt{X} ~.
\end{equs}
This together with \eref{eq:resultatdiffpetitere} implies that on $\tilde A$,
\begin{equs}
\sup_{\sigma\in [0, \lambda]} \|\hat Q(\sigma)-\hat Q(0)\|\leq c E^{\frac 1 {\ell_i}-\frac 1 {\ell_p}} +c G_E + c\sqrt{X}(1+E^{\frac 1 {\ell_p} -\frac 12 })+ c E^{\frac 1{2 \ell_p} - \frac 34} \eta_E~.
\end{equs}
But by \pref{prop:limitDomBoundedBelow}, the left-hand side is bounded below by $ C>0$. Thus,
\begin{equs}
cE^{\frac 1{2 \ell_p} - \frac 34} \eta_E \geq C 	-c E^{\frac 1 {\ell_i}-\frac 1 {\ell_p}} -c G_E -c\sqrt{X}(1+E^{\frac 1 {\ell_p} -\frac 12 })~.
\end{equs}
We next choose $\varepsilon > 0$ small enough so that for all $E$ large enough,
\begin{equs}
c E^{\frac 1 {\ell_i}-\frac 1 {\ell_p}} +c G_E + c\sqrt{\varepsilon}(1+E^{\frac 1 {\ell_p} -\frac 12 }) \leq \frac C 2~.
\end{equs}
Then, on $\tilde A$ and for $E$ large enough,  $X<\varepsilon$ implies $cE^{\frac 1{2 \ell_p} - \frac 34} \eta_E > C/2$, so that
\begin{equs}
\P_{z_0}(\tilde A \cap \{X < \varepsilon \}) \leq \P_{z_0}\left\{cE^{\frac 1{2 \ell_p} - \frac 34} \eta_E > C/2\right\} \leq e^{-cE^{\frac 32 - \frac 1{\ell_p}}}~.
\end{equs}
	Since ${\frac 32 - \frac 1{\ell_p}} \geq 1$, the proof is complete.
\end{proof}

By \lref{lem:BvarepsilonX}, and since
\begin{equs}
 \sum_{b\in \CB} \int_0^\lambda \gamma_b  \tilde p_b^2 d \sigma \geq X\inf_{b\in \CB}\gamma_b ~,
\end{equs}
we obtain \eref{eq:toshowpinningdom} (for some $\varepsilon>0$ possibly smaller than that of \lref{lem:BvarepsilonX}). Thus, the proof of \pref{p:tauxU} is complete.

\begin{acknowledge}This work was initiated in the pleasant atmosphere of the Fields
Institute, Toronto.
Financial support was kindly provided by the ERC Advanced grant  290843:BRIDGES (JPE),
the Swiss National Science Foundation Grant 165057 (NC),  the National Science Foundation (NSF)  grant DMS-1515712 (LRB), 
the Leverhulme Trust (MH), as well as the ERC Consolidator grant 615897:CRITICAL (MH).

\end{acknowledge}

\bibliographystyle{Martin}

\bibliography{refs}

\end{document}